\documentclass[11pt]{article}

\usepackage[margin=1in]{geometry}
\usepackage{amsmath,amssymb,amsthm,mathtools}
\usepackage{array}
\usepackage{tikz}
\usetikzlibrary{decorations.pathreplacing}
\usepackage[hidelinks]{hyperref}
\usepackage{microtype}
\hypersetup{
  pdftitle={Counting and Covering in Nearest-Neighbour Representations of Boolean Functions},
  pdfauthor={Martin Anthony},
  pdfkeywords={nearest-neighbour representation, Boolean function, growth function, covering design, symmetric function, average-case complexity}
}

\newtheorem{theorem}{Theorem}[section]
\newtheorem{lemma}[theorem]{Lemma}
\newtheorem{corollary}[theorem]{Corollary}
\newtheorem{proposition}[theorem]{Proposition}
\theoremstyle{definition}

\theoremstyle{remark}
\newtheorem{remark}[theorem]{Remark}
\newtheorem{example}[theorem]{Example}

\newcommand{\cube}{\{0,1\}^n}
\newcommand{\NN}{\operatorname{NN}}
\newcommand{\BNN}{\operatorname{BNN}}
\newcommand{\NNk}[1]{#1\text{-}\operatorname{NN}}
\newcommand{\kNN}{\NNk{k}}
\newcommand{\wt}[1]{\lvert #1\rvert}
\newcommand{\supp}{\operatorname{supp}}
\newcommand{\dist}{d_H}
\newcommand{\calF}{\mathcal F}
\newcommand{\calK}{\mathcal K}
\newcommand{\Cover}{\mathsf C}
\newcommand{\e}{\mathrm e}
\newcommand{\MAJ}{\mathrm{MAJ}}

\title{Counting and Covering in Nearest-Neighbour Representations\\
of Boolean Functions}
\author{Martin Anthony\\
Department of Mathematics\\
London School of Economics and Political Science, London, UK\\
\texttt{m.anthony@lse.ac.uk}}
\date{}

\begin{document}
\maketitle

\begin{abstract}
We study the number of prototypes needed to represent Boolean functions by
nearest-neighbour classification. There are two distinct settings: the
prototypes may be arbitrary points of Euclidean space, or they may
themselves be required to lie in the Boolean cube.

For unrestricted prototypes, we strengthen a known lower bound for almost
all Boolean functions. The bound applies simultaneously to nearest-neighbour
voting rules with any number of voting neighbours, and substantially narrows
the gap with the known general upper bound. We also obtain a VC-dimension
bound for classes with a bounded number of prototypes, and show that it is
sharp in order in dimensions at least four.

We then study Boolean prototypes, beginning with symmetric threshold
functions. A connection with covering designs expresses the minimum number
of prototypes at every threshold level exactly in terms of a covering number,
and leads to further exact results for related monotone functions, including
disjunctive extensions and a characterisation of when a representation with
a single negative prototype is possible.

For a uniformly random Boolean function, the Boolean nearest-neighbour
complexity, as a proportion of the cube, is asymptotically close either to
one half or to one, with explicit limiting probabilities. In particular,
almost every Boolean function requires at least approximately half as many
prototypes as there are points in the cube, and one half is the largest
proportion for which such a lower bound holds.

Finally, we consider arbitrary symmetric Boolean functions. Their Boolean
nearest-neighbour complexity is closely approximated by a weighted
vertex-cover problem on paths. As a consequence, a uniformly random symmetric
function typically requires prototypes amounting to $11/20$ of the cube.
This is much larger than the upper bounds known when the prototypes are
allowed to lie anywhere in Euclidean space.
\end{abstract}



\clearpage
\tableofcontents
\bigskip

\section{Introduction}
\label{sec:introduction}

A finite set of labelled prototypes classifies the points of the Boolean
cube by proximity, each point taking the label carried by a closest
prototype, and a Boolean function $f$ on $\{0,1\}^n$ is represented when
the classification agrees with it everywhere. The complexity is the least number of prototypes that will
serve: $\NN(f)$ when they may lie anywhere in $\mathbb R^n$, and $\BNN(f)$
when they are confined to $\{0,1\}^n$. Boolean representations use points of
the cube itself, as in the knowledge-compilation setting of \cite{CG25},
rather than auxiliary Euclidean points. Requiring the prototypes to be
Boolean can increase the minimum size exponentially. Every non-constant
linear threshold function satisfies $\NN(f)=2$, whereas some symmetric
threshold functions need exponentially many Boolean prototypes
\cite[Theorem~4]{HLT22}. We examine this discrepancy through counting
arguments and constraints on Boolean prototype sets.

The systematic study of these quantities is due to Hajnal, Liu, and
Tur\'an \cite{HLT06,HLT22}, who established the universal upper bound
$\NN(f)\leq(1+o(1))2^{n+2}/n$ together with the almost-all lower bound
$\kNN(f)>2^{n/2}/n$, the latter simultaneously for every number $k$ of
voting neighbours \cite[Theorems~5, 7, and~8]{HLT22}. Later work has
connected nearest-neighbour representations to Boolean circuits
\cite{DPR25,KSB24C}, examined the complexity and coordinate resolution of
representations of symmetric and threshold functions \cite{KSB23,KSB24N},
and placed Boolean representations within the knowledge-compilation map
\cite{CG25}.

We begin with a growth-function bound for real-prototype rules, valid for all
choices of $k$. It gives a VC-dimension upper bound of order $dm\log m$
for rules with at most $m$ prototypes in $\mathbb R^d$.
For $d\geq4$ and $m\geq3$, this is sharp in order already for $k=1$
(Corollary~\ref{cor:vc-sharp}). The matching lower bound combines a result
of Csik\'os, Mustafa, and Kupavskii~\cite{CMK19} on intersections of
half-spaces with the reflection construction of Gunn and Kuncheva~\cite{GK19}.
For each fixed $d\geq3$, this establishes the $O(m\log m)$ upper bound
envisaged in their concluding discussion.
Applied on the Boolean cube, the growth-function bound raises the almost-all
lower bound to more than $2^{n-1}/n^2$; for
$k=1$, the gap to the universal upper bound is then a factor $O(n)$.
Counting signed subsets of the cube directly gives, in the Boolean case, an
almost-all lower bound $(c_\ast-o(1))2^n$, where $H_2$ is the binary entropy
function, $H_2(c_\ast)+c_\ast=1$, and $c_\ast=0.227092\ldots$, again
simultaneously for every $k$.

Three constraints underlie the structural results. First, if two adjacent cube
points have different function values and one is a prototype, then so is the
other. Second, every shortest path from a point to a closest prototype consists
entirely of points with the same function value. This implies that every
connected component of the 0-points and of the 1-points contains a prototype.
Third, minimal positive points impose coordinatewise restrictions on their
closest positive prototypes, together with further Hamming-weight restrictions
in certain boundary cases.

The covering arguments rely mainly on the boundary restrictions, while the
probabilistic results use the connected-component constraint. Section~\ref{sec:monotone-extensions} combines these ideas in several
monotone constructions.

For $1\leq t\leq n$, let $TH_t^n$ denote the symmetric threshold function,
with $TH_t^n(x)=1$ precisely when $x$ has Hamming weight at least $t$. For
this function, the following covering bound is attained:
\[
 \BNN(TH_t^n)=\Cover(n,2r-1,r)+1,
 \qquad r=\min\{t,n-t+1\},
\]
where $\Cover(n,b,s)$ denotes the least number of $b$-subsets of $[n]$,
where $[n]=\{1,2,\ldots,n\}$, covering every $s$-subset. For threshold
levels near the middle of the cube, it is convenient to put $q=n-2r+1$, so
that $2r-1=n-q$. Complementation turns this covering problem into a uniform
hypergraph problem. For $q\geq1$, let $T(n,n-r,q)$ denote the least number of
$q$-subsets of $[n]$ such that every $(n-r)$-subset contains at least one of
them. Then
\[
 \Cover(n,n-q,r)=T(n,n-r,q).
\]
We use the covering identity for threshold functions to obtain uniform
exponential rates, formulas near the endpoint thresholds, and exact values near
majority. We then adapt the same structural ideas to obtain exact formulas for
several related monotone functions.

For a uniformly random Boolean function, the connected-component constraint
gives a sharp almost-all lower constant of $1/2$. More precisely, the
normalised Boolean nearest-neighbour complexity is asymptotically concentrated
near $1/2$ and $1$, with explicit limiting probabilities.

For symmetric functions, the Boolean prototype problem is reduced to a weighted
vertex-cover problem on paths, up to lower-order error terms. An explicit
family whose labels alternate between runs of three Hamming layers shows that
the uniform approximation error is optimal up to a logarithmic factor. For a uniformly
random symmetric function, we show that the Boolean nearest-neighbour complexity
is asymptotic to $(11/20)2^n$. Known real-prototype constructions require only
linear size with high probability \cite{KSB23}, so Boolean prototypes have
exponential cost in this model.

After the definitions in Section~\ref{sec:definitions},
Section~\ref{sec:capacity} gives the counting bounds, and
Sections~\ref{sec:boolean-structure}--\ref{sec:monotone-extensions} develop
the structural constraints, covering formulas, and monotone extensions.
Sections~\ref{sec:uniform-random-limit} and~\ref{sec:random-symmetric}
treat the two random-function models, with the deterministic approximation
for symmetric functions established in the latter section;
Section~\ref{sec:open-problems} states the open problems.

\section{Definitions, notation, and basic symmetries}
\label{sec:definitions}

Following \cite[Definition~1]{HLT22}, a nearest-neighbour representation of
$f:\cube\to\{0,1\}$ consists of disjoint sets $P,N\subseteq\mathbb R^n$
such that every positive point is strictly closer to some point of $P$ than
to every point of $N$, and every negative point is strictly closer to some
point of $N$ than to every point of $P$.  Its size is $\lvert P\rvert+
\lvert N\rvert$.  The minimum size is $\NN(f)$.  If $P,N\subseteq\cube$,
the representation is Boolean and its minimum size is $\BNN(f)$.  Distances
between points and real prototypes are Euclidean.  Throughout, prototype
collections are sets rather than multisets, so repeated copies of a prototype
are not permitted.

In a Boolean nearest-neighbour representation, every prototype carries its
own function value, since it is uniquely closest to itself. Thus, for a fixed
$f$, the prototype set $S=P\cup N$ determines the labelling, with
$P=S\cap f^{-1}(1)$ and $N=S\cap f^{-1}(0)$. Constant functions satisfy
$\NN(f)=\BNN(f)=1$.

In this nearest-neighbour model, ties are permitted only when all nearest
prototypes have the correct label. Thus a point may have several closest
prototypes, but not closest prototypes of both labels. Section~\ref{sec:capacity}
defines the voting and separation conventions for $k$-nearest-neighbour
representations.

Write $[n]=\{1,2,\ldots,n\}$.  For $x,y\in\cube$, let $\dist(x,y)$ be
Hamming distance and let $\wt{x}$ be Hamming weight.  Write
$\supp(x)=\{i\in[n]:x_i=1\}$. We write $x\leq y$ if $x_i\leq y_i$ for
every $i$; equivalently, $\supp(x)\subseteq\supp(y)$. For cube points,
squared Euclidean distance agrees with
Hamming distance, since
$\lVert x-y\rVert_2^2=\dist(x,y)$ for $x,y\in\cube$.

We call a Boolean function \emph{symmetric} when its value depends on its
point only through the Hamming weight.  The points of weight $j$, for
$0\leq j\leq n$, constitute the $j$th \emph{Hamming layer}.  For an integer $1\leq t\leq n$, the symmetric threshold
function is
\[
  TH_t^n(x)=1 \quad\Longleftrightarrow\quad \wt{x}\geq t.
\]
The majority function is $\MAJ_n=TH^n_{\lceil n/2\rceil}$.

Throughout, $\ln$ denotes the natural logarithm.  We write
\[
 H_2(u)=-u\log_2u-(1-u)\log_2(1-u), \qquad 0\leq u\leq1,
\]
for the binary entropy function, with the usual convention
$0\log_2 0=0$.

For $1\leq s\leq b\leq n$, let the covering number $\Cover(n,b,s)$ be the
least size of a family $\mathcal B\subseteq\binom{[n]}b$ such that every
member of $\binom{[n]}s$ is contained in some member of $\mathcal B$.

The notation is the standard one for covering designs \cite{GKP95}: $n$ is
the ground-set size, $b$ the block size, and $s$ the size of the sets to
be covered. Here $k$ is reserved for the number of neighbours, and $t$ for
the threshold.

The following basic symmetry observation will be used throughout.

\begin{lemma}\label{lem:nn-symmetries}
Let $T$ be a composition of coordinate permutations and coordinate
complementations of $\{0,1\}^n$. For every Boolean function $f$, the transformed
function $f\circ T$, and its negation $1-f\circ T$, have the same
nearest-neighbour and Boolean nearest-neighbour complexities as $f$:
\[
 \NN(f\circ T)=\NN(1-f\circ T)=\NN(f),\qquad
 \BNN(f\circ T)=\BNN(1-f\circ T)=\BNN(f).
\]
In particular, the dual function $f^*(x)=1-f(\mathbf1-x)$ satisfies
$\NN(f^*)=\NN(f)$ and $\BNN(f^*)=\BNN(f)$.
Using the affine extension of $T$, a representation $(P,N)$ of $f$ is
carried to $(T^{-1}P,T^{-1}N)$ for $f\circ T$ and to $(T^{-1}N,T^{-1}P)$
for $1-f\circ T$. The two prototype counts are preserved in the first
case and interchanged in the second.
\end{lemma}

\begin{proof}
The map $T$ extends to an isometry of $\mathbb R^n$ that maps the Boolean
cube bijectively to itself. Applying $T^{-1}$ to every prototype preserves
its distance from the corresponding transformed point, and therefore
transforms a representation of $f$ into one of $f\circ T$ with the same
size. Swapping the two prototype labels then represents $1-f\circ T$.
These operations preserve the strict cross-label inequalities for ordinary
nearest-neighbour representations. They are reversible, so the corresponding
minimum sizes are equal.
\end{proof}

\section{Counting bounds}
\label{sec:capacity}

When $k=1$, we use the nearest-neighbour definition in Section~\ref{sec:definitions}.
Thus several prototypes may be equally near a point, provided they all carry
its label.

For $2\leq k\leq\lvert P\rvert+\lvert N\rvert$, a point evaluates to one
exactly when at least half of the $k$ closest prototypes belong to $P$. Thus,
for even $k$, a tied vote counts as one. For $k\geq2$, we follow
\cite[Definition~2]{HLT22} in requiring that, at every point, the $k$th
closest prototype is strictly nearer than every prototype outside the closest
$k$. This condition is vacuous when $k$ is the total number of prototypes. We
write $\kNN(f)$ for the minimum size of a $k$-nearest-neighbour representation
of $f$, and $+\infty$ if there is none. When $k=1$, this means $\NN(f)$. For
$k\geq2$, these voting and separation conventions are also those of
\cite{DPR25}. When all prototypes lie in $\cube$, we call the representation
a Boolean $k$-nearest-neighbour representation.

\subsection{Removing distance ties}

For the counting arguments below, we use the same definitions for binary
labellings of finite sets $A\subseteq\mathbb R^d$, with prototypes in
$\mathbb R^d$ and the classification and separation conditions imposed
at every point of $A$.

The counting argument in the next subsection is simplest when, at each point of
$A$, the prototypes are strictly ordered by distance; the following lemma shows
that this may be assumed without loss of generality.

\begin{lemma}\label{lem:perturb}
Let $A\subseteq\mathbb R^d$ be finite and non-empty, and suppose that a binary labelling of $A$ is realised by a nearest-neighbour representation with $m$ real prototypes, or by a $k$-nearest-neighbour representation with $m$ real prototypes for a fixed $k\geq2$. We may perturb the prototypes so that:
\begin{enumerate}
  \item the induced labelling of $A$ is unchanged, and
  \item for every $x\in A$, all $m$ distances from $x$ are distinct, and the
  prototype locations remain distinct.
\end{enumerate}
\end{lemma}

\begin{proof}
Consider first a nearest-neighbour representation. If it has no prototype of one of the
two labels, any sufficiently small perturbation preserves its
constant labelling.  Suppose instead that both labels occur among the
prototypes.  By the definition of a nearest-neighbour representation, for every $x\in A$,
the minimum distance from $x$ to a prototype bearing its label is strictly
less than the minimum distance to a prototype bearing the other label.  Since
$A$ is finite, the difference between these two minimum distances has a
positive minimum over $x\in A$.
Every sufficiently small perturbation preserves the represented label at
every point of $A$.  We also take the perturbation neighbourhoods small
enough to be pairwise disjoint, so distinct prototypes remain distinct.

For a fixed $k$-nearest-neighbour representation with $2\leq k<m$, use instead the positive
minimum, over $A$, of the gap between the distance of the $k$th closest
prototype and that of the closest prototype outside the first $k$.  Every
sufficiently small perturbation preserves the set of the $k$ closest
prototypes at every point and hence preserves the represented labelling. When $k=m$, there is no outside prototype to separate from the first $k$; every
prototype participates in the vote, so perturbing the locations does not change
the induced labelling.

We now remove distance ties without changing the induced labelling. Fix one
prototype and hold the others fixed. For each point $x\in A$ and each fixed
prototype $p_j$, the moving prototype must avoid the sphere centred at $x$ with
radius $\lVert x-p_j\rVert_2$; points on this sphere are exactly those at the
same distance from $x$ as $p_j$. Thus the forbidden positions form a finite
union of spheres, whose complement is dense. We may therefore move the
prototype by an arbitrarily small amount that creates no new equality at any
point of $A$, while remaining within the margin that preserves the inequalities
already secured. Repeating this for the finitely many prototypes yields pairwise
distinct prototype distances at every $x\in A$. Since $A$ is non-empty,
the prototype locations are therefore distinct as well.
\end{proof}

\subsection{A growth-function bound}

Fix $d,m\geq2$. For a finite $A\subseteq\mathbb R^d$, let
$\mathcal L_{d,m}(A)$ be the set of binary labellings of $A$ realised by a
$k$-nearest-neighbour representation with at most $m$ real prototypes, for
at least one choice of $k$ (with $k=1$ meaning a nearest-neighbour
representation). Thus $\mathcal L_{d,m}(A)$ is the union of the labelling
classes obtained by allowing $k$ to range from $1$ to $m$. Let
$\Pi_{d,m}(\ell)$ be the maximum, over all $\ell$-point sets
$A\subseteq\mathbb R^d$, of $\lvert\mathcal L_{d,m}(A)\rvert$.

Inaba, Katoh, and Imai bound the number of Euclidean Voronoi partitions of a
finite point set by comparing, for each data point and each pair of prototype
locations, the two squared distances from the data point to those locations
\cite{IKI94}. The proof below uses the same parameter space and comparison
polynomials, with Warren's theorem supplying the sign-pattern bound. Hajnal,
Liu, and Tur\'an instead counted separately for each pair of prototypes
\cite[proof of Theorem~7]{HLT22}.

\begin{lemma}\label{lem:growth}
Let $d,m,\ell\geq2$ and suppose that
$\ell\binom m2\geq dm$.  Then
\[
 \Pi_{d,m}(\ell)
 \leq
 m2^m
 \left(\frac{4\e\ell(m-1)}{d}\right)^{dm}
 \leq
 m2^m
 \left(\frac{4\e\ell m}{d}\right)^{dm}.
\]
\end{lemma}

\begin{proof}
Fix $A=\{x_1,\ldots,x_\ell\}\subseteq\mathbb R^d$. On this finite set, we may
pad a representation having $r<m$ prototypes to exactly $m$ by adding $m-r$
distinct, arbitrarily labelled prototypes, each further from every point of
$A$ than every original prototype. Keep $k$ fixed. If $k<r$, the added
prototypes lie beyond all the original prototypes, so at every point of $A$
the $k$ nearest prototypes and their separation from the rest are unchanged.
If $k=r$, then before padding every original prototype participates in the vote.
After padding, the added prototypes lie strictly further from every point of
$A$ than the original prototypes do, so the $k=r$ closest prototypes are still
precisely the original prototypes. Thus the induced labelling of $A$ is
unchanged.
We may now restrict attention to representations with exactly $m$ prototypes.
To bound how many labellings of $A$ they can realise, choose a fixed ordering
$p_1,\ldots,p_m$ of the prototypes. Choosing an ordering may count the same
representation more than once, but that is harmless because we only need an
upper bound.

For every $x\in A$ and $i<j$, define
\[
 Q_{x,i,j}(p_1,\ldots,p_m)
   =\lVert x-p_i\rVert_2^2-\lVert x-p_j\rVert_2^2.
\]
There are
\[
 M=\ell\binom m2
\]
such polynomials, each of degree $2$ in the $D=dm$ real coordinates
of $p_1,\ldots,p_m$.  The hypothesis says that $M\geq D$.

After padding, Lemma~\ref{lem:perturb} shows that it suffices to consider
prototype configurations
in which, for every $x\in A$, the distances from $x$ to the prototypes are
all distinct. Following the parameterised-class framework of Goldberg and
Jerrum \cite{GJ95}, regard the prototype coordinates as the parameters.
Warren's Theorem~3~\cite{Warren68} bounds the number of consistent assignments
of non-zero signs to the $M$ comparison polynomials of
degree at most $2$ in $D$ variables by
\[
 \left(\frac{8\e M}{D}\right)^D.
\]
In the present case this is
\[
 \left(
   \frac{4\e\ell(m-1)}{d}
 \right)^{dm}
 \leq
 \left(
   \frac{4\e\ell m}{d}
 \right)^{dm}.
\]

Each such assignment determines the complete ordering of the prototypes by
distance at every point of $A$.  Once the $m$ prototype labels and the value
of $k$ are specified, the induced labelling of $A$ is determined.  There are
at most $2^m$ prototype labellings and at most $m$ choices of $k$, proving
the claim.
\end{proof}

Let $V_{d,m}$ denote the VC dimension of the class counted by
$\Pi_{d,m}$, so it is the largest $\ell$ for which
$\Pi_{d,m}(\ell)=2^\ell$. Shattering is understood with the preceding
separation conventions satisfied at every point of the sample.

\begin{corollary}\label{cor:vc-dimension}
For $d,m\geq2$,
\[
 V_{d,m}\leq 2dm\log_2(16\e m^2)<dm(4\log_2m+11).
\]
In particular, $V_{d,m}=O(dm\log_2 m)$, with an absolute implicit constant.
\end{corollary}

\begin{proof}
Suppose that $\Pi_{d,m}(\ell)=2^\ell$. If $\ell<dm$ there is nothing to
prove. Otherwise $\ell\binom m2\geq dm$, so Lemma~\ref{lem:growth}
applies, and since $m2^m\leq4^m\leq2^{dm}$ it gives
\[
 2^\ell\leq\left(\frac{8\e m}{d}\right)^{dm}\ell^{dm}.
\]
Writing $D=dm$ and using $u\leq2^u$ with $u=\ell/(2D)$,
\[
 \ell^{D}=(2D)^{D}\left(\frac{\ell}{2D}\right)^{D}\leq(2D)^{D}2^{\ell/2}.
\]
Combining the two displays gives $2^{\ell/2}\leq(16\e m^2)^{dm}$, that
is, $\ell\leq2dm\log_2(16\e m^2)$.
\end{proof}

For fixed-size nearest-neighbour rules with $k=1$, Gunn and Kuncheva
\cite{GK19} give a VC-dimension upper bound of order $dm^2\log_2(dm)$ for
$d\geq3$ (and $m\log_2m$ for $d=2$).
Corollary~\ref{cor:vc-dimension} applies to the larger class obtained by
allowing all admissible values of $k$. For each fixed $d\geq3$, it improves
the upper-bound dependence on the prototype budget from
$O(m^2\log_2m)$ to $O(m\log_2m)$. They also give the lower bound
\[
 V^{(1)}_{d,m}\geq dm+2-d
 \qquad(d\geq2,\ m\geq3),
\]
where $V^{(1)}_{d,m}$ denotes the VC dimension for $k=1$. Since the class
counted by $V_{d,m}$ includes the case $k=1$, we have
$V_{d,m}\geq V^{(1)}_{d,m}$. For each fixed $m\geq3$, these bounds give
$V_{d,m}=\Theta(d)$ as $d\to\infty$.

For $d\geq4$, a stronger lower bound follows by combining two known
results.

\begin{corollary}[Sharp VC-dimension order]\label{cor:vc-sharp}
For $d\geq4$ and $m\geq3$,
\[
 V^{(1)}_{d,m}=\Theta(dm\log_2m),
 \qquad
 V_{d,m}=\Theta(dm\log_2m),
\]
with absolute implicit constants. Thus the upper bound of
Corollary~\ref{cor:vc-dimension} is sharp in order already for $k=1$.
\end{corollary}

\begin{proof}
Csik\'os, Mustafa, and Kupavskii~\cite[Theorem~1(b)]{CMK19}
show that the class of intersections of at most $r$ half-spaces in
$\mathbb R^d$ has VC dimension $\Omega(dr\log_2r)$, for $r\geq2$,
with an absolute implicit constant. Such intersections can be realised
on finite samples by nearest-neighbour rules using at most $r+1$
prototypes, through the reflection construction of Gunn and
Kuncheva~\cite[Lemma~1 and Appendix~A]{GK19}.

To check the separation convention, move the boundary hyperplanes of the
closed half-spaces slightly outwards so that no sample point lies on a
boundary and the induced labelling is unchanged. Such a choice is possible
because the sample is finite: boundary points move strictly inside, while
each point outside a half-space has a positive distance from its boundary.
For a non-constant
labelling, choose an interior point $p$ of the intersection, label it
positive, and label its reflections in the distinct boundary hyperplanes
negative. A sample point is closer to $p$ than to every reflected
prototype precisely when it lies in the intersection; outside it, some
reflected prototype is strictly closer than $p$. Thus the required
cross-label inequalities hold. Constant labellings require only one
prototype. The construction does not require the intersection to be bounded.

Taking $r=m-1$, for $m\geq3$ we obtain
\[
 V^{(1)}_{d,m}=\Omega\bigl(d(m-1)\log_2(m-1)\bigr)
 =\Omega(dm\log_2m).
\]
Since $V^{(1)}_{d,m}\leq V_{d,m}$, Corollary~\ref{cor:vc-dimension}
supplies the matching upper bound for both classes.
\end{proof}

For each fixed $d\geq3$, Corollary~\ref{cor:vc-dimension} establishes the
$O(m\log m)$ upper bound envisaged in the concluding discussion of Gunn
and Kuncheva~\cite{GK19}.
Corollary~\ref{cor:vc-sharp} determines the order for every fixed $d\geq4$.
The order in dimensions two and three is not determined here. In those
dimensions, intersections of at most $r$ half-spaces have VC dimension
$O(r)$; see the bounds reviewed in~\cite[Section~1.1]{CMK19}.
Thus the reflection argument through a single intersection cannot supply
the logarithmic factor. These bounds do not give upper bounds for general
nearest-neighbour rules, whose positive decision regions need not be convex.

\subsection{Application to Boolean functions}

Let $\calK_{n,m}$ denote the set of functions on $\cube$ that have a
$k$-nearest-neighbour representation with at most $m$ real prototypes for at
least one $k$. Each member of
$\calK_{n,m}$ is a realisable binary labelling of the
fixed $2^n$-point set $\cube$. Since $\Pi_{n,m}(2^n)$ is the maximum number
of realisable labellings over all $2^n$-point subsets of $\mathbb R^n$, we
have
\[
 \lvert\calK_{n,m}\rvert\leq\Pi_{n,m}(2^n).
\]
For $n,m\geq2$, the hypothesis of Lemma~\ref{lem:growth} holds, since
$2^{n-1}\geq n$ and $m-1\geq1$, and hence
\[
 2^n\binom m2=2^{n-1}m(m-1)\geq nm.
\]
Applying the lemma with $d=n$ and $\ell=2^n$ therefore gives
\begin{equation}\label{eq:cube-count}
 \lvert\calK_{n,m}\rvert
 \leq
 m2^m
 \left(\frac{4\e\,2^n m}{n}\right)^{nm}.
\end{equation}

\begin{theorem}[Almost-all lower bound]\label{thm:almost-all}
The proportion of $n$-variable Boolean functions $f$ for which there is some
$k$ satisfying
\[
  \kNN(f)
  \leq
  \frac{2^{n-1}}{n^2}
\]
tends to zero as $n\to\infty$.  Consequently, almost every $f$ satisfies,
simultaneously for every $k$,
\[
  \kNN(f)
  >
  \frac{2^{n-1}}{n^2}.
\]
In particular, almost every $f$ satisfies $\NN(f)>2^{n-1}/n^2$.
\end{theorem}

\begin{proof}
Put
\[
 m=\left\lfloor\frac{2^{n-1}}{n^2}\right\rfloor.
\]
Since prototype counts are integer-valued, $\calK_{n,m}$ is exactly the set
of functions occurring in the first assertion of the theorem.  For all
sufficiently large $n$, we have $m\geq2$, so \eqref{eq:cube-count} applies.
Taking base-2 logarithms gives
\[
 \log_2\lvert\calK_{n,m}\rvert
 \leq
 \log_2m+m+
 nm\left(n+\log_2m-\log_2n+\log_2(4\e)\right).
\]
Since $m\leq2^{n-1}/n^2$,
\[
 \log_2m\leq n-1-2\log_2n,
 \qquad
 nm\leq\frac{2^{n-1}}n.
\]
The bracket is at most $2n-3\log_2 n+O(1)$, and the terms $\log_2m+m$ are
absorbed into $O(2^n/n)$. We therefore obtain
\[
 \log_2\lvert\calK_{n,m}\rvert
 \leq
 2^n-\frac32\,\frac{2^n\log_2n}{n}
 +O\left(\frac{2^n}{n}\right).
\]
Since there are $2^{2^n}$ Boolean functions on $\cube$, the proportion
having a $k$-nearest-neighbour representation with at most $m$ prototypes
for some $k$ is
\[
 \frac{\lvert\calK_{n,m}\rvert}{2^{2^n}}
 \leq
 2^{-\frac32\,2^n\log_2n/n+O(2^n/n)},
\]
which tends to zero.
\end{proof}

\subsection{A direct count for Boolean prototypes}

The preceding argument accommodates real prototypes. When prototypes are
required to be Boolean, a direct count gives a substantially stronger lower
bound.

\begin{proposition}
\label{prop:boolean-count}
Put $Q=2^n$, and let $F$ be a uniformly random Boolean function on $\cube$.
For every integer $m\geq0$, let $p(m)$ be the probability that $F$ has a
Boolean $k$-nearest-neighbour representation using at most $m$ prototypes
for some $k\geq1$. Then
\[
 p(m)\leq 2^{-Q}\sum_{j=1}^m j\binom Qj2^j.
\]
\end{proposition}

\begin{proof}
There are $\binom Qj2^j$ signed sets of $j$ Boolean prototypes. Fix such a
set with $j\geq1$. For $k=1$, at each point either all the nearest
prototypes have the same label, in which case that label is forced, or
nearest prototypes of both labels occur, in which case the strict
cross-label inequalities fail. Thus the signed set represents at most one
Boolean function when $k=1$.

For each fixed $k=1,\ldots,j$, the signed set determines at most one Boolean
function. Indeed, either the $k$-nearest-neighbour rule is well defined at
every cube point, in which case the induced label at every point is determined,
or the required separation condition fails somewhere, in which case this signed
set and this value of $k$ represent no function. Consequently, the signed set
represents at most $j$ Boolean functions as $k$ ranges from $1$ to $j$.
Summing over the signed sets and dividing by the $2^Q$ Boolean functions
gives the result. The sum begins at $j=1$ because the empty signed set
represents no function.
\end{proof}

\begin{corollary}
\label{cor:almost-all-BNN}
For $0<c<1/2$, put $\delta(c)=1-H_2(c)-c$ and $m=\lfloor cQ\rfloor$.
Then
\[
 p(m)\leq m\,2^{-Q\delta(c)}.
\]
Let $c_\ast=0.227092\ldots$ denote the unique root in $(0,1/2)$ of
\[
 H_2(c_\ast)+c_\ast=1.
\]
For every $\varepsilon>0$, with probability tending to one as $n\to\infty$,
every Boolean $k$-nearest-neighbour representation of $F$, for every $k\geq1$,
uses more than $(c_\ast-\varepsilon)2^n$ prototypes.
\end{corollary}

\begin{proof}
The standard entropy bound for binomial sums gives
\[
 \sum_{j=0}^{\lfloor cQ\rfloor}\binom Qj\leq2^{QH_2(c)}
 \qquad(0\leq c\leq1/2).
\]
For $1\leq j\leq m$, we have $j\leq m$ and $2^j\leq2^{cQ}$, so
Proposition~\ref{prop:boolean-count} yields
\[
 p(m)
 \leq m\,2^{-Q+cQ}\sum_{j=0}^m\binom Qj
 \leq m\,2^{-Q\delta(c)}.
\]
The function $c\mapsto H_2(c)+c$ is continuous and strictly increasing on
$[0,1/2]$, with endpoint values $0$ and $3/2$, so $c_\ast$ exists and is
unique. We may assume that $0<\varepsilon<c_\ast$, since otherwise the
conclusion is immediate. Take $c=c_\ast-\varepsilon$. Then $\delta(c)>0$,
and the displayed bound tends to zero because $m\leq Q$.
\end{proof}

\section{Structural constraints}
\label{sec:boolean-structure}

From this section onwards, ``representation'' means a Boolean
nearest-neighbour representation in the sense of the first paragraph of
Section~\ref{sec:definitions}. Voting rules with $k\geq2$ reappear only where
explicitly stated.

\subsection{Bichromatic components and monochromatic shortest paths}

Let $Q_n$ denote the graph of the Boolean cube. For $f:\cube\to\{0,1\}$,
let $H_f$ be the graph with vertex set $\cube$ whose edges are precisely
the cube edges $xy$ for which $f(x)\ne f(y)$. A connected component of
$H_f$ will be called a \emph{bichromatic component}. Here ``bichromatic''
refers to the edges: every edge of $H_f$ joins vertices with different
function values.

A set of cube vertices is called \emph{monochromatic} if $f$ takes the
same value at every vertex in the set. The connected components of the
subgraphs induced by $f^{-1}(0)$ and $f^{-1}(1)$ are therefore called
monochromatic components.

\begin{proposition}[Closure across a change of label]\label{prop:edge-closure}
If $(P,N)$ is a Boolean nearest-neighbour representation of $f$, its
prototype set $S=P\cup N$ is a union of connected components of $H_f$.
\end{proposition}
\begin{proof}
Let $C$ be a connected component of
$H_f$ that meets $S$, and choose $x_0\in C\cap S$. For an arbitrary
$y\in C$, take a path $x_0,x_1,\ldots,x_t=y$ in $C$. Suppose that
$x_i\in S$. Since $x_ix_{i+1}$ is an edge of $H_f$, the prototype $x_i$
has the opposite label from $x_{i+1}$ and is at distance 1 from it.
Correct classification of $x_{i+1}$ therefore requires a prototype of its
own label at distance strictly less than 1. Hamming distances are integral,
so this prototype is $x_{i+1}$ itself. Thus $x_{i+1}\in S$, and induction
gives $y\in S$. Since $y$ was arbitrary, $C\subseteq S$. Hence every
component meeting $S$ is contained in $S$, so $S$ is a union of connected
components of $H_f$.
\end{proof}

For parity and its complement, $H_f=Q_n$ is connected, so the existence
of any prototype and Proposition~\ref{prop:edge-closure} force the entire
cube to be selected. Thus their Boolean complexity is $2^n$, as proved
in \cite[Proposition~3(b)]{HLT22}.

The following lemma records the shortest-path form of the argument used by
DiCicco, Podolskii, and Reichman \cite{DPR25}. They state the resulting
lower bound for the connected components of $f^{-1}(1)$; the same proof,
with the labels interchanged, gives the corresponding statement for
$f^{-1}(0)$. The isolated-vertex consequence, again for both labels, is
also noted by \v{C}epek and Gli\v{s}i\'c \cite{CG25}.

\begin{lemma}\label{lem:monochromatic-shortest-paths}
Let $(P,N)$ be a representation of $f:\cube\to\{0,1\}$, and let $x\in\cube$.
If $p$ is a closest prototype to $x$, every shortest cube path from $x$
to $p$ stays in the function label class of $x$. Consequently every
monochromatic connected component contains a prototype, and every vertex
whose neighbours all have the opposite function value is a prototype.
\end{lemma}
\begin{proof}
Because $x$ is correctly classified, every closest prototype to $x$ has
label $f(x)$. Let $p$ be one such prototype. If a vertex $y$ on a shortest
path from $x$ to $p$ had the opposite label, then
\[
 d_H(x,p)=d_H(x,y)+d_H(y,p).
\]
Since $y$ has the opposite function value from $x$, the prototype $p$, which has label $f(x)$, has the wrong label for $y$. Correct classification of $y$ therefore gives a prototype $q$ with label $f(y)$, opposite to $f(x)$, such that
\[
 d_H(y,q)<d_H(y,p).
\]
The triangle inequality would then give
\[
 d_H(x,q)\leq d_H(x,y)+d_H(y,q)
 <d_H(x,y)+d_H(y,p)=d_H(x,p),
\]
contradicting the choice of $p$. Hence every shortest path from $x$ to
$p$ is monochromatic.

Now let $C$ be a monochromatic component and choose $x\in C$. If $p$ is
a closest prototype to $x$, then a shortest path from $x$ to $p$ is
monochromatic, so $p\in C$. Thus every monochromatic component contains
a prototype. In particular, a vertex whose neighbours all have the opposite
function value forms a monochromatic component on its own, so is a prototype.
\end{proof}

Once Lemma~\ref{lem:monochromatic-shortest-paths} identifies a prototype in a
monochromatic component, Proposition~\ref{prop:edge-closure} forces every
vertex in its connected component of $H_f$ to be a prototype as well.

\subsection{Boundary containment}
\label{subsec:boundary-constraints}

Recall that $x\leq y$ means $x_i\leq y_i$ for every coordinate $i$.
Minimal positive points and maximal negative points are, respectively,
minimal elements of $f^{-1}(1)$ and maximal elements of $f^{-1}(0)$ under
this order. We say that a minimal positive point $a$ is \emph{exposed} if
it has a maximal negative lower neighbour; that is, if there is a maximal
negative point $b<a$ with $d_H(a,b)=1$. Write $A(f)$ for the set of
exposed minimal positive points of $f$.

The next lemma is the bridge to the covering arguments used later: it shows
that suitable positive prototypes must cover the exposed positive boundary.

\begin{lemma}\label{lem:boundary}
Let $f$ be a Boolean function on $\cube$ (not necessarily monotone), and let $(P,N)$ be a
Boolean nearest-neighbour representation of $f$.
\begin{enumerate}
  \item If $a$ is a minimal positive point and $p\in P$ is closest to $a$
  among the positive prototypes, then $a\leq p$.
  \item If $b$ is a maximal negative point and $q\in N$ is closest to $b$
  among the negative prototypes, then $q\leq b$.
  \item If $a$ is exposed, $b$ is a lower neighbour of $a$ that is a
  maximal negative point of $f$, and $q\in N$ is closest to $b$
  among the negative prototypes, then every $p\in P$ closest to $a$
  among the positive prototypes satisfies
  \[
     a\leq p
     \quad\text{and}\quad
     \wt{p}\leq2\wt{a}-\wt{q}-1\leq2\wt{a}-1.
  \]
\end{enumerate}
\end{lemma}

Parts 1 and 2 appear for symmetric threshold functions in
\cite[proof of Theorem~4(c)]{HLT22}; their proofs require only minimal
positivity or maximal negativity, so they apply to arbitrary Boolean
functions. In the proof of their even-majority lower bound, DiCicco,
Podolskii, and Reichman show that every point $a$ of weight $n/2$ has a
positive prototype $p\geq a$ with $\wt p\leq n-1$
\cite{DPR25}. This is the corresponding case of Part
3 with the term $-\wt q$ discarded.

\begin{proof}
For the first claim, suppose that $a_i=1$ and $p_i=0$.  Let $b$ be
obtained from $a$ by changing coordinate $i$ to zero.  Since $a$ is a minimal
positive point, $b$ is negative.  Let $q\in N$ be closest
to $b$ among the negative prototypes.  Since $b$ is negative,
$\dist(b,q)<\dist(b,p)$.  Hence
\[
 \dist(a,p)=\dist(b,p)+1>\dist(b,q)+1.
\]
But $a$ is positive, so
\[
 \dist(a,p)<\dist(a,q)\leq\dist(b,q)+1,
\]
a contradiction.  Thus $a\leq p$.

Apply the first claim to $f^*(x)=1-f(\mathbf1-x)$, using
Lemma~\ref{lem:nn-symmetries}. A maximal negative point $b$ of $f$
becomes the minimal positive point $\mathbf1-b$ of $f^*$, and a closest
negative prototype $q$ becomes the closest positive prototype
$\mathbf1-q$. Thus $\mathbf1-b\leq\mathbf1-q$, or $q\leq b$, proving
the second claim.

For the third claim, the second claim applies to $b$ and gives
$q\leq b\leq a$, and therefore
$\dist(a,q)=\wt{a}-\wt{q}$.  If $p$ is closest to $a$ among the positive
prototypes, then
\[
 \dist(a,p)<\dist(a,q)=\wt{a}-\wt{q}.
\]
Distances are integral, so $\dist(a,p)\leq\wt{a}-\wt{q}-1$.  By the first
claim, $a\leq p$. Since $a$ and $p$ are Boolean points,
\[
  \wt{p}=\wt{a}+\dist(a,p)\leq2\wt{a}-\wt{q}-1\leq2\wt{a}-1.
\]
\end{proof}

The term $-\wt q$ can strengthen the bound for a particular
representation, but the optimal threshold construction in
Theorem~\ref{thm:covering} uses $q=0^n$. The exposure hypothesis cannot
be omitted, as Remark~\ref{rem:sharpness} shows.

\begin{remark}\label{rem:sharpness}
For $f(x)=x_1$ on $\{0,1\}^3$, the sets $P=\{110,101,111\}$ and
$N=\{010,001,011\}$ form a Boolean nearest-neighbour representation.
The unique minimal positive point $a=100$ is not exposed: its only lower
neighbour is $000$, which is not maximal negative. Every positive
prototype $p\geq a$ has weight at least two, whereas
$2\wt a-1=1$. Thus the conclusion of Part 3 of Lemma~\ref{lem:boundary} can fail when the minimal
positive point is not exposed.
\end{remark}

\section{Symmetric threshold functions and covering designs}
\label{sec:thresholds}

For $TH_t^n$, the boundary-containment lemma turns the Boolean
nearest-neighbour problem into a covering-design problem. The minimal positive
points have weight $t$, the maximal negative points have weight $t-1$, and
every minimal positive point is exposed.

Recall that $\Cover(n,b,s)$ is the least number of $b$-subsets of $[n]$
required to cover every $s$-subset. The exact connection is as follows.

\begin{theorem}[Exact covering formula]\label{thm:covering}
For every $1\leq t\leq n$, let
\[
  r=\min\{t,n-t+1\}.
\]
Then
\[
   \BNN(TH_t^n)=\Cover(n,2r-1,r)+1.
\]
\end{theorem}

\begin{proof}
First suppose $t\leq(n+1)/2$, so that $2t-1\leq n$.

For the lower bound, let $(P,N)$ be any Boolean nearest-neighbour
representation of $TH_t^n$. Part 3 of
Lemma~\ref{lem:boundary} shows that the supports of
the positive prototypes of weight at most $2t-1$ cover every $t$-subset
of $[n]$.  Enlarge each such support, if necessary, to a
$(2t-1)$-subset (possible since $2t-1\leq n$), and discard any duplicate
enlarged blocks. These enlarged sets are used only for the covering argument,
not as replacements for the original prototypes. Removing repetitions does
not affect which $t$-subsets are covered. The resulting family is therefore
an $(n,2t-1,t)$ covering of size at most $\lvert P\rvert$, so
\[
  \lvert P\rvert\geq\Cover(n,2t-1,t).
\]
The function is non-constant, hence there is at least one negative
prototype.  Thus
\[
  \BNN(TH_t^n)\geq\Cover(n,2t-1,t)+1.
\]

For the upper bound, take an optimal covering
$\mathcal B\subseteq\binom{[n]}{2t-1}$.  For each $B\in\mathcal B$, use
its incidence vector $\mathbf 1_B$ as a positive prototype, and use
$0^n$ as the sole negative prototype. Let $x\in\cube$ be arbitrary and
write $S=\supp(x)$ and $w=\wt{x}=\lvert S\rvert$, so that
$TH_t^n(x)=0$ exactly when $w\leq t-1$ and $TH_t^n(x)=1$ exactly when
$w\geq t$. Since $\lvert B\rvert=2t-1$,
\[
 \dist(x,\mathbf1_B)
 =\lvert S\mathbin{\triangle}B\rvert
 =w+2t-1-2\lvert S\cap B\rvert,
\]
whereas $\dist(x,0^n)=w$. Hence
\[
 \dist(x,\mathbf1_B)-\dist(x,0^n)
 =2t-1-2\lvert S\cap B\rvert.
\]

If $w\leq t-1$, then $\lvert S\cap B\rvert\leq t-1$, so this difference is
at least $1$. Thus $0^n$ is closer to $x$ than every positive prototype.

Now suppose that $w\geq t$. Choose a $t$-subset $T\subseteq S$. Since
$\mathcal B$ covers every $t$-subset, some $B\in\mathcal B$ contains $T$.
It follows that $\lvert S\cap B\rvert\geq t$, so
\[
 \dist(x,\mathbf1_B)-\dist(x,0^n)
 \leq-1.
\]
Thus the positive prototype $\mathbf1_B$ is closer to $x$ than $0^n$.

This is a representation with
$\Cover(n,2t-1,t)+1$ prototypes.

For general $t$, the dual of $TH_t^n$ is $TH_{n-t+1}^n$.
Lemma~\ref{lem:nn-symmetries} therefore allows us to replace $t$ by
$r=\min\{t,n-t+1\}$, completing the proof.
\end{proof}

For example, $\{1,2,3\}$, $\{1,2,4\}$, and $\{1,3,4\}$ cover all
pairs of $\{1,2,3,4\}$. Any two triples omit distinct points $i$ and $j$, so the pair $\{i,j\}$ is
not covered. Hence at least three triples are needed. Thus $\Cover(4,3,2)=3$ and $\BNN(TH_2^4)=4$.

\begin{corollary}\label{cor:one-side}
For every $1\leq t\leq n$, there is an optimal Boolean nearest-neighbour
representation of $TH_t^n$ with a single negative prototype when
$t\leq(n+1)/2$, and one with a single positive prototype when
$t\geq(n+1)/2$. If $n$ is odd and $t=(n+1)/2$, then
$\BNN(TH_t^n)=2$.
\end{corollary}

\begin{proof}
For $t\leq(n+1)/2$, use the construction in the proof of
Theorem~\ref{thm:covering}. For $t\geq(n+1)/2$, use the duality between
$TH_t^n$ and $TH_{n-t+1}^n$ from Lemma~\ref{lem:nn-symmetries}, which
interchanges the positive and negative prototype counts. In the central case, $2t-1=n$ and
$\Cover(n,n,t)=1$; the prototypes $0^n$ and $1^n$ therefore give an
optimal representation.
\end{proof}

Section~\ref{subsec:one-negative} later characterises the monotone Boolean
functions that admit a Boolean nearest-neighbour representation with a single
negative prototype.

The exact formula immediately recovers and sharpens the known exponential
lower bounds for symmetric threshold functions. Hajnal, Liu, and Tur\'an
proved the qualitative exponential lower bound for $TH_{\lfloor n/3\rfloor}^n$
by a boundary-counting argument \cite[Theorem~4(c)]{HLT22}, and noted that
the argument extends to further threshold levels. DiCicco, Podolskii, and
Reichman state the corresponding estimate for every $t$ \cite{DPR25}.
Together with the duality between $TH_t^n$ and $TH_{n-t+1}^n$ from
Lemma~\ref{lem:nn-symmetries}, this gives
\[
 \BNN(TH_t^n)\geq\frac{\binom nr}{\binom{2r}r},
 \qquad r=\min\{t,n-t+1\}.
\]
By Lemma~\ref{lem:standard-cover-estimates},
\[
 \Cover(n,2r-1,r)\geq
 \frac{\binom nr}{\binom{2r-1}r}
 =2\,\frac{\binom nr}{\binom{2r}r},
\]
so the exact formula improves this lower bound by a factor of two, apart from
the additional $+1$. At $t=\lfloor n/3\rfloor$, standard binomial estimates
applied to the lower-bound expression show that it has exponential rate
$\log_2 3-4/3$. Theorem~\ref{thm:covering}, together with the covering upper
bound below, gives a corresponding upper bound with the same exponential rate.
In this sense, the exact covering formula determines the exponential rate at
this threshold.

\subsection{Exponential rates}

\begin{lemma}[Standard covering estimates]\label{lem:standard-cover-estimates}
For all $1\leq s\leq b\leq n$,
\[
 \frac{\binom ns}{\binom bs}
 \leq \Cover(n,b,s)
 \leq
 \frac{\binom ns}{\binom bs}\left(1+\ln\binom bs\right).
\]
\end{lemma}

\begin{proof}
The lower bound is the elementary counting bound: each $b$-subset contains
$\binom bs$ $s$-subsets. For $n>b>s\geq2$, the upper bound is the
classical estimate of Erd\H os and Spencer \cite[Theorem~13.4]{ES74}, in
whose notation $M(n,b,s)=\Cover(n,b,s)$. The remaining cases follow
directly from
\[
 \Cover(n,b,b)=\binom nb,\qquad
 \Cover(n,n,s)=1,\qquad
 \Cover(n,b,1)=\left\lceil\frac nb\right\rceil.
\]
\end{proof}

Combining Theorem~\ref{thm:covering} with these standard covering
estimates gives the exponential rate of $\BNN(TH_t^n)$ throughout the
threshold range.

\begin{corollary}\label{cor:entropy}
There is an absolute constant $c$ such that for every $n\geq2$ and every
$1\leq t\leq n$, with $r=\min\{t,n-t+1\}$,
\[
 \Bigl\lvert
 \log_2\BNN(TH_t^n)-\bigl(nH_2(r/n)-2r\bigr)
 \Bigr\rvert
 \leq c\log_2n.
\]
For $\alpha\in[0,1]$, define
\[
 \beta=\min\{\alpha,1-\alpha\},\qquad
 \rho(\alpha)=H_2(\beta)-2\beta.
\]
If $t=t(n)$ satisfies $t/n\to\alpha$, in particular if
$t=\max\{1,\lfloor\alpha n\rfloor\}$, then
\[
 \log_2\BNN(TH_t^n)=\bigl(\rho(\alpha)+o(1)\bigr)n.
\]
\end{corollary}

\begin{proof}
Put
\[
 L=\frac{\binom nr}{\binom{2r-1}r}.
\]
Theorem~\ref{thm:covering} and Lemma~\ref{lem:standard-cover-estimates} give
\[
 L+1\leq\BNN(TH_t^n)\leq1+L(1+n\ln2),
\]
since $\binom{2r-1}r\leq2^{2r-1}\leq2^n$. As $L\geq1$, it follows that
$\log_2\BNN(TH_t^n)=\log_2L+O(\log_2 n)$.
The method-of-types estimates for binomial coefficients
\cite[Proposition~1.5]{PW25}, whose upper-bound form includes the entropy
estimate used earlier, give
\[
 \log_2\binom nr=nH_2(r/n)+O(\log_2 n),
 \qquad
 \log_2\binom{2r}r=2r+O(\log_2 n).
\]
Together with $\binom{2r-1}r=\tfrac12\binom{2r}r$, this gives
\[
 \log_2\binom{2r-1}r=2r+O(\log_2 n),
\]
uniformly for $1\leq r\leq(n+1)/2$, with constants independent of $r$.
It follows that
\[
 \log_2L=nH_2(r/n)-2r+O(\log_2n).
\]
Combining this with the preceding bound on
$\log_2\BNN(TH_t^n)$ proves the first assertion. If $t/n\to\alpha$,
then $r/n\to\beta$, and hence
\[
 \frac1n\log_2\BNN(TH_t^n)
 =H_2(r/n)-\frac{2r}{n}+O\left(\frac{\log_2n}{n}\right)
 \longrightarrow H_2(\beta)-2\beta.
\]
This is the second assertion.
\end{proof}

\begin{remark}
For $t=\lfloor n/3\rfloor$, Corollary~\ref{cor:entropy} gives
\[
 \log_2\BNN(TH_{\lfloor n/3\rfloor}^n)
 =\left(\log_2 3-\frac43\right)n+O(\log_2 n)
 \approx0.251629\,n+O(\log_2 n).
\]
Equivalently,
\[
 \BNN(TH_{\lfloor n/3\rfloor}^n)
 =\left(\frac{3}{2^{4/3}}\right)^n n^{O(1)}
 =(1.19055\ldots)^n n^{O(1)}.
\]
\end{remark}

The rate $\rho$ is symmetric about $1/2$. On $[0,1/2]$,
$\rho(\alpha)=H_2(\alpha)-2\alpha$ is strictly concave, vanishes at the two
endpoints, and is positive in between. Its zeros on $[0,1]$ are therefore
precisely $0$, $1/2$, and $1$. On $(0,1/2)$ its derivative is
$\rho'(\alpha)=\log_2((1-\alpha)/\alpha)-2$, so its maximum on
$(0,1/2)$ occurs at $1/5$, and by symmetry the other maximum occurs at
$4/5$, with common value $\log_2(5/4)$. In particular,
Corollary~\ref{cor:entropy} gives
$\BNN(TH_t^n)=(5/4)^n n^{O(1)}$ when
$t=n/5+O(1)$ or $t=4n/5+O(1)$.

\subsection{Threshold levels near the endpoints}

When $r$ is fixed, Corollary~\ref{cor:entropy} gives only
$\log_2\BNN(TH_t^n)=O(\log_2n)$. We now use results on covering designs
to obtain exact values for $r=1,2$, an asymptotic formula for every
fixed $r$, and further exact values when suitable designs exist.

\begin{corollary}
\label{cor:fixed-r}
Let $r\geq1$ be a fixed integer, let $n\geq2r-1$, and let $t\in\{r,n-r+1\}$.
\begin{enumerate}
  \item For $r=1$ and every $n\geq1$,
  $\BNN(TH_t^n)=n+1$.
  \item For $r=2$ and every $n\geq3$,
  \[
   \BNN(TH_t^n)
   =\left\lceil\frac n3\left\lceil\frac{n-1}2\right\rceil\right\rceil+1.
  \]
  \item For every fixed $r$, as $n\to\infty$,
  \[
   \BNN(TH_t^n)
   =(1+o(1))\,\frac{\binom nr}{\binom{2r-1}r}.
  \]
\end{enumerate}
\end{corollary}

\begin{proof}
For $t\in\{r,n-r+1\}$, Theorem~\ref{thm:covering} gives
$\BNN(TH_t^n)=\Cover(n,2r-1,r)+1$. Part 1 follows from
$\Cover(n,1,1)=n$. For Part 2, Fort and Hedlund \cite[Section~2]{FH58}
prove $\Cover(n,3,2)=\lceil(n/3)\lceil(n-1)/2\rceil\rceil$ for every
$n\geq3$. For fixed $r\geq2$, R\"odl's theorem \cite{Rodl85} gives
$\Cover(n,2r-1,r)=(1+o(1))\binom nr/\binom{2r-1}r$; the additive
$1$ is absorbed as $n\to\infty$. The case $r=1$ of Part 3 follows
from Part 1.
\end{proof}

A Steiner system $S(r,2r-1,n)$ covers every $r$-set exactly once, so it
has exactly $\binom nr/\binom{2r-1}r$ blocks and meets the lower bound in
Lemma~\ref{lem:standard-cover-estimates} with equality. For fixed
$r\geq2$, Keevash's existence theorem for designs \cite{Keevash14} gives
such systems for every sufficiently large admissible $n$, where
admissibility means that
\[
 \binom{2r-1-i}{r-i}\ \text{divides}\ \binom{n-i}{r-i}
 \qquad(0\leq i<r).
\]
For these values of $n$,
\[
 \BNN(TH_t^n)=\frac{\binom nr}{\binom{2r-1}r}+1,
 \qquad t\in\{r,n-r+1\}.
\]

\subsection{The central window}

Recall that $r=\min\{t,n-t+1\}$. Put $q=n-2r+1$, so that
\[
 r=\frac{n-q+1}{2},\qquad 2r-1=n-q.
\]
The two threshold levels corresponding to $r$, namely $t=r$ and
$t=n-r+1$, are therefore
\[
 t\in\left\{\frac{n-q+1}{2},\frac{n+q+1}{2}\right\}.
\]
The integers $n$ and $q$ have opposite parity. Fixing $q$ fixes the
displacement from the central threshold levels; for odd $n$, this
displacement from majority is $q/2$, with $q=0$ giving majority.

The covering formula recovers the known majority values. For odd $n$,
$\Cover(n,n,(n+1)/2)=1$, so $\BNN(\MAJ_n)=2$. For even $n$, an
$(n-1)$-block is determined by the single point it omits, and a family of
such blocks covers every $n/2$-set precisely when at least $n/2+1$
distinct points are omitted by its members. Hence
\[
 \BNN(\MAJ_n)=\Cover(n,n-1,n/2)+1=n/2+2.
\]
The odd value and the construction giving the even upper bound are due to
Hajnal, Liu, and Tur\'an \cite{HLT22}. DiCicco, Podolskii, and Reichman
proved the lower bound in the even case \cite{DPR25}; the omitted-point calculation is
their argument in covering notation.

For the remaining central cases, it is useful to express the covering
problem in the language of hypergraphs. A hypergraph $\mathcal H$ on $[n]$
is a family of subsets of $[n]$, called edges, and it is $q$-uniform if
every edge has cardinality $q$. A vertex cover, also called a transversal, is a set of vertices meeting
every edge, and $\tau(\mathcal H)$ denotes the minimum cardinality of such
a set.

Since $2r-1=n-q$, the covering number in the exact formula is
$\Cover(n,n-q,r)$. Each of its blocks is the complement of a $q$-set, so a
family of such blocks corresponds to a $q$-uniform hypergraph formed by
their complements. An $r$-set is covered by a block precisely when it is
disjoint from the corresponding edge. The following standard complementation
observation, often used to pass between covering numbers and Tur\'an numbers
\cite{GKP95}, gives the resulting hypergraph formulation.

\begin{lemma}\label{lem:complement}
Let $r\geq1$ and $1\leq q\leq n-r$. Then $\Cover(n,n-q,r)$ is the minimum
number of edges in a $q$-uniform hypergraph $\mathcal H$ on $[n]$ satisfying
$\tau(\mathcal H)\geq r+1$.
\end{lemma}

\begin{proof}
Given a family $\mathcal B$ of $(n-q)$-blocks, define
\[
 \mathcal H=\{[n]\setminus B:B\in\mathcal B\}.
\]
Complementation is a bijection, so $\mathcal H$ is $q$-uniform and
$\lvert\mathcal H\rvert=\lvert\mathcal B\rvert$. An $r$-set $T$ satisfies
$T\subseteq B$ if and only if $T$ is disjoint from the edge
$[n]\setminus B$. Hence the blocks cover every $r$-set if and only if
every $r$-set is disjoint from some edge, that is, if and only if no
$r$-set meets every edge. Any vertex cover with fewer than $r$ vertices can
be enlarged to an $r$-set, so this is equivalent to
$\tau(\mathcal H)\geq r+1$.
\end{proof}

We shall use the lemma only when $q=n-2r+1\geq1$, so that
$n-q=2r-1$. In this case $n\geq2r$, and hence $q\leq n-r$.

\begin{remark}
An $r$-set is a vertex cover precisely when its complement contains no
edge. Hence $\tau(\mathcal H)\geq r+1$ is equivalent to requiring every
$(n-r)$-set to contain an edge. In the Tur\'an-system notation introduced in Section~1,
Lemma~\ref{lem:complement} therefore gives
\[
 \Cover(n,n-q,r)=T(n,n-r,q).
\]
\end{remark}

\begin{corollary}[The case $q=2$]\label{cor:near-majority}
Let $n\geq3$ be odd. Then
\[
 \BNN\left(TH_{(n-1)/2}^{\,n}\right)
 =
 \BNN\left(TH_{(n+3)/2}^{\,n}\right)
 =\frac{n+5}{2}.
\]
\end{corollary}

\begin{proof}
Put $r=(n-1)/2$, so that $n=2r+1$ and $q=2$. By
Lemma~\ref{lem:complement}, it suffices to show that the
minimum number of edges of a graph $G$ on $2r+1$ vertices with
$\tau(G)\geq r+1$ is $r+2$.

Suppose that $G$ has $m\leq r+1$ edges. If they are pairwise disjoint,
then $m\leq r$, and choosing one endpoint of each edge gives a vertex
cover of size at most $r$. Otherwise, two edges meet. Their common
endpoint, together with one endpoint from each of the remaining edges,
is a vertex cover of size at most $1+(m-2)\leq r$. Thus any graph with
$\tau(G)\geq r+1$ has at least $r+2$ edges.

Conversely, take a triangle and a matching of size $r-1$ on the remaining
$2r-2$ vertices. A triangle requires two vertices in any vertex cover,
while the $r-1$ disjoint matching edges require one vertex each. Since
these components are vertex-disjoint, $\tau(G)=2+(r-1)=r+1$.

Thus the covering number is $r+2$. Theorem~\ref{thm:covering} therefore gives
\[
 \BNN\left(TH_{(n-1)/2}^{\,n}\right)=r+3=\frac{n+5}{2},
\]
and the value for the dual threshold follows from
Lemma~\ref{lem:nn-symmetries}.
\end{proof}

The majority calculations above give the cases $q=0,1$, and
Corollary~\ref{cor:near-majority} gives $q=2$. The case $q=3$ can also be
evaluated exactly.

\begin{corollary}[The case $q=3$]\label{cor:q3}
Let $n\geq4$ be even and put $r=(n-2)/2$, so that
$q=n-2r+1=3$. Let $t\in\{r,\,n-r+1\}$. Then
\[
 \Cover(n,n-3,r)=n
\]
and hence
\[
 \BNN(TH_t^n)=n+1.
\]
\end{corollary}

\begin{proof}
By Lemma~\ref{lem:complement}, $\Cover(n,n-3,r)$ is the least number of
edges of a $3$-uniform hypergraph $\mathcal H$ on $[n]$ with
$\tau(\mathcal H)\geq r+1=n/2$.
The $q=3$ case of the theorem of Chv\'atal and McDiarmid
\cite{CM92} states that every $3$-uniform hypergraph on $n$
vertices satisfies
\[
 \tau(\mathcal H)\leq\frac{n+\lvert\mathcal H\rvert}{4}.
\]
It follows that $\lvert\mathcal H\rvert\geq n$.

For the matching upper bound, note first that $\tau(K_4^{(3)})=2$, where
$K_4^{(3)}$ is the complete $3$-uniform hypergraph on four vertices: two
vertices meet every $3$-subset of a $4$-set, while a single vertex
misses the $3$-subset on the remaining three vertices.

Next, let $H_6$
be the $3$-uniform hypergraph on $\mathbb Z_6$ with the six edges
$\{x,x+1,x+3\}$ for $x\in\mathbb Z_6$. The three pairs in the base edge
$\{0,1,3\}$ realise the differences $1$, $2$, and $3$ up to sign modulo $6$.
Since every pair of distinct vertices of $\mathbb Z_6$ has one of these
differences, and the edge set consists of all translates of the base edge,
every pair of vertices lies in some edge. Since each vertex has degree
three, a pair of vertices meets at most $3+3-1=5$ of the six edges, and so
misses one.  Hence $\tau(H_6)\geq3$.  Conversely, $\{0,1,2\}$ meets
every edge, so $\tau(H_6)=3$.

Vertex-cover numbers add over disjoint
components.  For $n\equiv0\pmod4$, take $n/4$
pairwise disjoint copies of $K_4^{(3)}$: together they have $n$ edges
and $\tau=n/2$. For $n\equiv2\pmod4$, take $(n-6)/4$ pairwise
disjoint copies of $K_4^{(3)}$ and one further disjoint copy of $H_6$:
again there are $n$ edges, and
$\tau=2\cdot\tfrac{n-6}4+3=\tfrac n2$.

This gives a $3$-uniform hypergraph with $n$ edges and
$\tau(\mathcal H)=n/2$, so the lower bound is attained. The formula for
$\BNN$ follows from Theorem~\ref{thm:covering}.
\end{proof}

We next give bounds when $q$, and hence the displacement from the central
threshold levels, is fixed.

\begin{corollary}\label{cor:fixed-q}
Let $1\leq q\leq n-1$, and suppose that $n$ and $q$ have opposite parity.
Put
\[
 r=\frac{n-q+1}{2},
\]
and let
\[
 t\in\left\{\frac{n-q+1}{2},\frac{n+q+1}{2}\right\}.
\]
Then
\[
 r+2
 \leq
 \BNN(TH_t^n)
 \leq
 1+2^q(1+n\ln2).
\]
If $q\geq2$, then the lower bound can be strengthened to
\[
 \BNN(TH_t^n)
 \geq
 1+\left\lceil
 \frac{\lfloor3q/2\rfloor(r+1)-n}{\lfloor q/2\rfloor}
 \right\rceil.
\]
For every fixed $q\geq1$, as $n\to\infty$ through integers of parity
opposite to $q$,
\[
 \BNN(TH_t^n)=\Theta(n).
\]
\end{corollary}

\begin{proof}
Let $\mathcal H$ be any hypergraph in the characterisation of
Lemma~\ref{lem:complement}, and put $m=\lvert\mathcal H\rvert$. Choosing
one vertex from each edge gives a vertex cover of size at most $m$. Hence
\[
 r+1\leq\tau(\mathcal H)\leq m.
\]
Every such hypergraph therefore has at least $r+1$ edges, and hence
Lemma~\ref{lem:complement} gives $\Cover(n,n-q,r)\geq r+1$. Together with
Theorem~\ref{thm:covering}, this gives the first lower bound.

For $q\geq2$, the theorem of Chv\'atal and McDiarmid
\cite{CM92} states that every $q$-uniform hypergraph
$\mathcal H$ on a vertex set of size $n$ with $m$ edges satisfies
\[
 \tau(\mathcal H)
 \leq
 \frac{n+\lfloor q/2\rfloor m}{\lfloor3q/2\rfloor}.
\]
Since the hypergraph here has $\tau(\mathcal H)\geq r+1$, rearranging and
using the integrality of $m$ gives the strengthened lower bound.
For $q=2$ the bound reads
$\BNN(TH_t^n)\geq(n+5)/2$, and for $q=3$ it reads
$\BNN(TH_t^n)\geq n+1$. By Corollary~\ref{cor:near-majority} and
Corollary~\ref{cor:q3}, respectively, both bounds are exact.

For the upper bound, Lemma~\ref{lem:standard-cover-estimates} gives
\[
 \Cover(n,b,s)
 \leq
 \frac{\binom ns}{\binom bs}
 \left(1+\ln\binom bs\right).
\]
Apply this with $b=n-q=2r-1$ and $s=r$. The ratio in this bound satisfies
\[
 \frac{\binom nr}{\binom{2r-1}r}
 =
 \prod_{i=0}^{q-1}\frac{2r+i}{r+i}
 \leq 2^q,
\]
since each of the $q$ factors is at most $2$. Also,
$\ln\binom{n-q}r\leq n\ln2$. It follows that
\[
 \Cover(n,n-q,r)\leq2^q(1+n\ln2).
\]
By Theorem~\ref{thm:covering},
$\BNN(TH_t^n)=\Cover(n,n-q,r)+1$, so
\[
 \BNN(TH_t^n)\leq1+2^q(1+n\ln2).
\]

If $q\geq1$ is fixed, then $r=(n-q+1)/2=\Theta(n)$. The lower and upper
bounds above therefore give $\BNN(TH_t^n)=\Theta(n)$.
\end{proof}

The following table summarises the main consequences of this section. Here
$r=\min\{t,n-t+1\}$ and $q=n-2r+1$; in the final row,
$r=(n-q+1)/2$ and $t\in\{(n-q+1)/2,(n+q+1)/2\}$.

\begin{center}
\renewcommand{\arraystretch}{1.25}
\small
\begin{tabular}{>{\raggedright\arraybackslash}p{0.16\textwidth}>{\raggedright\arraybackslash}p{0.48\textwidth}>{\raggedright\arraybackslash}p{0.27\textwidth}}
\textbf{Range} & \textbf{Result obtained} & \textbf{Main input} \\
\hline
$1\leq t\leq n$
& $\BNN(TH_t^n)=\Cover(n,2r-1,r)+1$
& boundary containment \\
$t/n\to\alpha$
& $\log_2\BNN(TH_t^n)=(\rho(\alpha)+o(1))n$
& standard covering estimates \\
$r$ fixed
& $\BNN(TH_t^n)\sim\binom nr/\binom{2r-1}r$
& R\"odl's theorem \cite{Rodl85} \\
$q=0$
& $\BNN(TH_{(n+1)/2}^n)=2$ for odd $n$
& two antipodal prototypes \\
$q=1$
& $\BNN(TH_{n/2}^n)=\BNN(TH_{n/2+1}^n)=n/2+2$ for even~$n$
& majority calculation \\
$q=2$
& $\BNN(TH_{(n-1)/2}^n)=\BNN(TH_{(n+3)/2}^n)=(n+5)/2$ for odd $n$
& graph vertex covers \\
$q=3$
& $\BNN(TH_{(n-2)/2}^n)=\BNN(TH_{(n+4)/2}^n)=n+1$ for even~$n$
& Chv\'atal--McDiarmid \cite{CM92} \\
$q\geq1$ fixed
& $\BNN(TH_t^n)=\Theta(n)$, with explicit bounds in Corollary~\ref{cor:fixed-q}
& covering estimates and Chv\'atal--McDiarmid \cite{CM92}
\end{tabular}
\end{center}

\section{Monotone extensions and one-negative representations}
\label{sec:monotone-extensions}

This section studies monotone functions that have especially simple
nearest-neighbour representations. We first give a general covering lower bound,
then apply it to functions built from smaller Boolean functions. The section
ends with a characterisation of the monotone functions that can be represented
using just one negative prototype.

Recall from Section~\ref{subsec:boundary-constraints} that $A(f)$ is the set
of exposed minimal positive points of $f$. Thus $a\in A(f)$ is a minimal point
of $f^{-1}(1)$ that has a maximal negative lower neighbour.

For $A,C\subseteq[n]$, we say that $C$ \emph{covers} $A$ if $A\subseteq C$,
and that $C$ covers $a\in\cube$ if it contains $\supp(a)$. For a
non-constant monotone Boolean function $f$, let $c(f)$ be the least size
of a family $\mathcal B$ of subsets of $[n]$ such that, for every exposed
minimal positive point $a\in A(f)$, some $B\in\mathcal B$ satisfies
$\supp(a)\subseteq B$ and $|B|\leq2\wt a-1$. In particular,
$c(f)=0$ when $A(f)=\varnothing$. Recall that
$f^*(x)=1-f(\mathbf1-x)$ is the dual monotone function.

\begin{proposition}[Covering lower bound for monotone functions]\label{prop:monotone-cover}
Every Boolean nearest-neighbour representation $(P,N)$ of a non-constant
monotone Boolean function $f$ satisfies
\[
 |P|\geq\max\{1,c(f)\},\qquad
 |N|\geq\max\{1,c(f^*)\}.
\]
Consequently,
\[
 \BNN(f)\geq\max\{1,c(f)\}+\max\{1,c(f^*)\}
 \geq\max\{1,c(f)\}+1.
\]
\end{proposition}

\begin{proof}
For every $a\in A(f)$, Part 3 of Lemma~\ref{lem:boundary} supplies a
positive prototype $p$ such that $\supp(a)\subseteq\supp(p)$ and
$\wt p\leq2\wt a-1$. Thus the positive supports form a family satisfying
the conditions in the definition of $c(f)$, so $|P|\geq c(f)$. Non-constancy also gives
$|P|\geq1$. By Lemma~\ref{lem:nn-symmetries}, complementing every
coordinate and swapping the labels gives a representation of $f^*$ whose
positive prototypes correspond to $N$. Applying the same bound to this
representation gives $|N|\geq\max\{1,c(f^*)\}$. Adding the inequalities
and minimising over representations proves the bound for $\BNN(f)$.
\end{proof}

If $A(f)$ and $A(f^*)$ are non-empty, then $c(f)\geq1$ and
$c(f^*)\geq1$, so Proposition~\ref{prop:monotone-cover} gives
\[
 \BNN(f)\geq c(f)+c(f^*).
\]
For a non-constant function, every representation has at least one positive
and at least one negative prototype, which accounts for the $1$ in the weaker
bound
\[
 \BNN(f)\geq\max\{1,c(f)\}+1.
\]
This $+1$ is best possible in general: for $t\leq(n+1)/2$, the optimal
covering representation of $TH_t^n$ uses $0^n$ as its only negative prototype,
by Theorem~\ref{thm:covering} and Corollary~\ref{cor:one-side}.

\subsection{Disjunctive lifts}
\label{subsec:disjunctive-lifts}

We next consider Boolean functions on $\{0,1\}^h\times\{0,1\}^m$, writing
a point as $(z,x)$ with $z\in\{0,1\}^h$ and $x\in\{0,1\}^m$. We call $x$
the core part of the point and $z$ the auxiliary part; correspondingly, the
$x$-coordinates are the core coordinates and the $z$-coordinates are the
auxiliary coordinates. Given a Boolean function $f$ on $\{0,1\}^m$, its
$h$-fold disjunctive lift is $G_h(z,x)=\operatorname{OR}_h(z)\vee f(x)$.
Here $f$ is the core function. Write $e_j$ for the $j$th unit vector of
$\{0,1\}^h$.

The following theorem treats the disjunctive lift of an arbitrary non-constant
monotone core function.

\begin{theorem}[Disjunctive lifts of monotone functions]\label{thm:general-disjunctive-lift}
Let $f$ be a non-constant monotone Boolean function on $\{0,1\}^m$, let
$h\geq1$, and put
\[
 G_h(z,x)=\operatorname{OR}_h(z)\vee f(x).
\]
Every Boolean nearest-neighbour representation $(P,N)$ of $G_h$ satisfies
\[
 |P|+|N|\geq\max\{1,c(f)\}+(h+1)|N|.
\]
If $f$ has an optimal representation with exactly one negative prototype,
then
\[
 \BNN(G_h)\leq\BNN(f)+h.
\]
In particular, if
\[
 \BNN(f)=\max\{1,c(f)\}+1,
\]
so that the lower bound from Proposition~\ref{prop:monotone-cover} is attained
using only the positive-side covering parameter $c(f)$ together with one
negative prototype, then
\[
 \BNN(G_h)=\BNN(f)+h,
\]
and every optimal representation of $G_h$ has exactly one negative
prototype.
\end{theorem}

\begin{proof}
Every negative prototype of $G_h$ has the form $(0^h,y)$, where
$f(y)=0$. For each $(0^h,y)\in N$,
Proposition~\ref{prop:edge-closure} forces the $h$ positive prototypes
$(e_j,y)$, $j\in[h]$. These give $h|N|$ distinct positive prototypes.

Let
\[
 P_+=\{(z,y)\in P:f(y)=1\}
\]
be the set of positive prototypes whose core part $y$ is itself positive for
$f$.
We claim that $|P_+|\geq c(f)$. Let $a\in A(f)$, and choose a maximal
negative lower neighbour $b$ of $a$. Then $(0^h,a)$ is a minimal positive
point of $G_h$, and $(0^h,b)$ is a lower neighbour. Moreover, $(0^h,b)$ is
maximal negative: changing any auxiliary coordinate from $0$ to $1$ makes
$\operatorname{OR}_h(z)=1$, while changing any zero core coordinate of $b$
to $1$ makes $f$ positive because $b$ is maximal negative. Hence $(0^h,a)$
is exposed. Part 3 of Lemma~\ref{lem:boundary} therefore supplies a positive
prototype $(z,y)$ such that
\[
 a\leq y
 \qquad\text{and}\qquad
 \wt z+\wt y\leq2\wt a-1.
\]
By monotonicity, $f(y)=1$, and hence $(z,y)\in P_+$. For each $a\in A(f)$,
we have therefore found a prototype $(z,y)\in P_+$ such that
\[
 \supp(a)\subseteq\supp(y)
\]
and
\[
 |\supp(y)|=\wt y\leq \wt z+\wt y\leq 2\wt a-1.
\]
Hence the family of core supports
\[
 \{\supp(y):(z,y)\in P_+\}
\]
satisfies the covering conditions in the definition of $c(f)$. Its size is
at most $|P_+|$, so
\[
 c(f)\leq |P_+|,
\]
proving the claim.

We also have $|P_+|\geq1$. Indeed, choose any minimal positive point $a$
of $f$ and a closest positive prototype $(z,y)$ to $(0^h,a)$. Part 1 of
Lemma~\ref{lem:boundary} gives $a\leq y$, so monotonicity gives $f(y)=1$.

Consequently,
\[
 |P_+|\geq\max\{1,c(f)\}.
\]
The $h|N|$ forced positive prototypes described above all have core part
$y$ with $f(y)=0$, so none belongs to $P_+$. Therefore
\[
 |P|\geq\max\{1,c(f)\}+h|N|,
\]
which proves the first assertion.

For the upper bound in the case where $f$ has an optimal representation with
exactly one negative prototype, let $(P_0,\{q\})$ be such an optimal
representation. Use
\[
 \{(0^h,p):p\in P_0\}\cup\{(e_j,q):j\in[h]\}
\]
as the positive prototypes for $G_h$, and use $(0^h,q)$ as the sole
negative prototype. On the face $z=0^h$, the embedded prototypes preserve
all distances from the original representation of $f$. If $f(x)=0$, then
$(0^h,q)$ is closer to $(0^h,x)$ than every embedded positive prototype,
and it is exactly one unit closer than each additional prototype
$(e_j,q)$. If $f(x)=1$, some embedded positive prototype is closer than
$(0^h,q)$, and hence also closer than every $(e_j,q)$. Thus every point on
this face is classified correctly.

If $z\neq0^h$, choose $j$ such that $z_j=1$. Then
\[
 \dist((z,x),(e_j,q))=\dist((z,x),(0^h,q))-1.
\]
Thus a positive prototype is strictly closer than the sole negative
prototype, so every point with non-zero auxiliary part is classified
positively. This gives a representation with $\BNN(f)+h$ prototypes, and
therefore
\[
 \BNN(G_h)\leq\BNN(f)+h.
\]

Finally, suppose that
\[
 \BNN(f)=\max\{1,c(f)\}+1.
\]
Proposition~\ref{prop:monotone-cover} states that every
representation $(P',N')$ of $f$ satisfies
\[
 |P'|\geq\max\{1,c(f)\}
 \qquad\text{and}\qquad
 |N'|\geq1.
\]
Hence equality in the covering lower bound forces an optimal
representation of $f$ to have exactly one negative prototype. The
preceding construction therefore gives
\[
 \BNN(G_h)\leq\BNN(f)+h
 =\max\{1,c(f)\}+h+1.
\]

Conversely, $G_h$ is non-constant, so every representation of $G_h$ has
$|N|\geq1$. The lower bound proved above then gives
\[
 |P|+|N|
 \geq\max\{1,c(f)\}+(h+1)|N|
 \geq\max\{1,c(f)\}+h+1
 =\BNN(f)+h.
\]
It follows that $\BNN(G_h)=\BNN(f)+h$. Moreover, any representation with
at least two negative prototypes has size at least
\[
 \max\{1,c(f)\}+2(h+1)>\BNN(G_h).
\]
Hence every optimal representation of $G_h$ has exactly one negative
prototype.
\end{proof}

We now specialise to symmetric threshold cores. For $h,m\geq1$ and
$1\leq r\leq m$, put
\[
 F_{h,m,r}(z,x)=\operatorname{OR}_h(z)\vee TH_r^m(x).
\]
Equivalently, $F_{h,m,r}(z,x)=1$ precisely when
$r\sum_{j=1}^h z_j+\sum_{i=1}^m x_i\geq r$.

\begin{corollary}[Disjunctive threshold lifts]\label{cor:threshold-lift}
Let $h,m\geq1$ and $r\geq1$ with $2r-1\leq m$.
Every Boolean nearest-neighbour representation $(P,N)$ of $F_{h,m,r}$ satisfies
\[
 |P|+|N|\geq\Cover(m,2r-1,r)+(h+1)|N|.
\]
Moreover,
\[
 \BNN(F_{h,m,r})=\Cover(m,2r-1,r)+h+1
 =\BNN(TH_r^m)+h.
\]
Every optimal representation in this range has the sole negative
prototype $0^{h+m}$ and contains the $h$ positive prototypes
$(e_j,0^m)$, $j\in[h]$, supported on the individual auxiliary coordinates.
Moreover, if $\mathcal B$ is an optimal $(m,2r-1,r)$ covering, then an
optimal representation is obtained by taking $0^{h+m}$ as the negative
prototype and
\[
 \{(0^h,\mathbf1_B):B\in\mathcal B\}
 \cup
 \{(e_j,0^m):j\in[h]\}
\]
as the positive prototypes.
\end{corollary}

\begin{proof}
For $f=TH_r^m$ with $2r-1\leq m$, the exposed minimal positive supports
are precisely the $r$-subsets of $[m]$. Enlarging any smaller covering block
to size $2r-1$ shows that
\[
 c(f)=\Cover(m,2r-1,r).
\]
By Theorem~\ref{thm:covering},
\[
 \BNN(f)=c(f)+1.
\]
Hence Theorem~\ref{thm:general-disjunctive-lift} gives
\[
 \BNN(F_{h,m,r})
 =\BNN(f)+h
 =\Cover(m,2r-1,r)+h+1.
\]
It also follows from Theorem~\ref{thm:general-disjunctive-lift} that every
optimal representation has exactly one negative prototype. The explicit
optimal representation is obtained by applying the construction in that
theorem to the covering representation of $TH_r^m$ from
Theorem~\ref{thm:covering}.

For any optimal representation, write $N=\{q\}$. By Part 2 of
Lemma~\ref{lem:boundary}, the unique negative prototype lies below every
maximal negative point $(0^h,\mathbf1_B)$ with
$B\in\binom{[m]}{r-1}$. Their supports have empty intersection, because
every core coordinate is omitted by at least one such $B$. Hence
$q=0^{h+m}$. Proposition~\ref{prop:edge-closure} then forces all the
positive prototypes $(e_j,0^m)$, $j\in[h]$.
\end{proof}

The dual of $\operatorname{OR}_h\vee TH_r^m$ is
$\operatorname{AND}_h\wedge TH_{m-r+1}^m$. Applying
Lemma~\ref{lem:nn-symmetries} to Corollary~\ref{cor:threshold-lift}
therefore gives
\[
 \BNN(\operatorname{AND}_h\wedge TH_t^m)=\BNN(TH_t^m)+h
 \qquad\text{when }t\geq(m+1)/2.
\]

The conjunction core is itself a symmetric threshold function, since
\[
 \operatorname{AND}_m=TH_m^m.
\]
However, the corresponding parameter choice $r=m$ lies outside the range
$2r-1\leq m$ of Corollary~\ref{cor:threshold-lift} except when $m=1$.
We can nevertheless determine the complexity of its disjunctive lift exactly.

\begin{theorem}[Disjunctive lifts of conjunctions]\label{thm:orand}
For every $h,m\geq1$,
\[
 \BNN\bigl(\operatorname{OR}_h(z)\vee\operatorname{AND}_m(x)\bigr)
 =(h+1)\left\lceil\frac m{h+1}\right\rceil+1.
\]
\end{theorem}

\begin{proof}
We first bound the size of a representation $(P,N)$ using the missing core coordinates of its negative prototypes, and then give matching constructions in the two ranges $m\geq h+1$ and $m<h+1$. Write $a=(0^h,1^m)$. The point
$a$ is minimal positive, so by Part 1 of Lemma~\ref{lem:boundary},
every closest positive prototype at $a$ contains all core
coordinates. Among positive prototypes whose core coordinates are all $1$,
let $s$ be the minimum number of auxiliary coordinates equal to $1$.
The minimum Hamming distance from $a$ to a positive prototype is therefore
exactly $s$.

Every negative prototype has auxiliary part $0^h$. Writing $E$ for the
non-empty set of core coordinates at which it is zero, it therefore has the
form $(0^h,\mathbf1_{[m]\setminus E})$. Its Hamming distance from $a$ is
$|E|$. Since $a$ is positive and the minimum distance from $a$ to a positive
prototype is $s$, every negative prototype must be at distance greater than
$s$. Hence $|E|\geq s+1$.

For each $i\in[m]$, let
\[
 b_i=(0^h,\mathbf1_{[m]\setminus\{i\}}).
\]
This is a maximal negative point. By the definition of $s$, there is a
closest positive prototype $p$ to $a$ at distance $s$. Part 1 of
Lemma~\ref{lem:boundary} gives $a\leq p$; since $a=(0^h,1^m)$, the core part
of $p$ is $1^m$. Thus $p$ differs from $b_i$ in the same $s$ auxiliary
coordinates in which it differs from $a$, and also in core coordinate $i$.
Hence
\[
 \dist(p,b_i)=s+1.
\]
Since $b_i$ is negative, some negative prototype must be at distance at most
$s$ from $b_i$. For each $i\in[m]$, choose a closest negative prototype $q_i$
to $b_i$. Then $\dist(b_i,q_i)\leq s$, and Part 2 of
Lemma~\ref{lem:boundary} gives $q_i\leq b_i$. Let $E_i$ be the set of zero
core coordinates of $q_i$. Since $q_i\leq b_i$ and $(b_i)_i=0$, we have
$i\in E_i$. The distance from $q_i$ to $b_i$ is $|E_i|-1$, so
$|E_i|-1\leq s$.
On the other hand, the preceding argument at $a$ gives $|E_i|\geq s+1$.
Therefore $|E_i|=s+1$. Since $i\in E_i$ for every $i\in[m]$, the missing
sets of the negative prototypes cover $[m]$. Consequently
\[
 |N|\geq\left\lceil\frac m{s+1}\right\rceil
 \geq\left\lceil\frac m{h+1}\right\rceil.
\]
Each negative prototype $(0^h,y)$ forces the $h$ positive prototypes
\[
 (e_j,y),\qquad j\in[h],
\]
by Proposition~\ref{prop:edge-closure}. These $h|N|$ prototypes are distinct,
and each has the same non-full core part $y$ as the negative prototype that
forces it.

However, $a=(0^h,1^m)$ is a minimal positive point. By Part 1 of
Lemma~\ref{lem:boundary}, any closest positive prototype to $a$ must contain
all core coordinates, and hence must have core part $1^m$. None of the
$h|N|$ forced positive prototypes has this property. Therefore at least one
further positive prototype is required.

Thus
\[
 |P|+|N|\geq(h+1)|N|+1,
\]
which gives the required lower bound.

We now prove the matching upper bound.

Suppose first that $m\geq h+1$. Choose
$\ell=\lceil m/(h+1)\rceil$ distinct $(h+1)$-subsets of $[m]$
whose union is $[m]$. Such a family can be obtained by partitioning
$[m]$ into groups of size $h+1$, except possibly the final group,
and enlarging that group if necessary. For each chosen set $E$, use
\[
 q_E=(0^h,\mathbf1_{[m]\setminus E})\in N,
 \qquad (e_j,\mathbf1_{[m]\setminus E})\in P
 \quad(j\in[h]),
\]
and also take
\[
 p_*=(1^h,1^m)
\]
as a positive prototype.

Every positive point with non-zero auxiliary part is classified correctly:
if $z_j=1$, choose a closest negative prototype $q_E$; then its positive
counterpart $(e_j,\mathbf1_{[m]\setminus E})$ is one unit closer. The only
positive point with auxiliary part $0^h$ is
\[
 a=(0^h,1^m).
\]
At $a$, the positive prototype $p_*$ has distance $h$, whereas every negative
prototype $q_E$ has distance $|E|=h+1$. Hence $a$ is also classified
correctly.

It remains to classify a negative point
$u=(0^h,\mathbf1_{[m]\setminus T})$ with $T\ne\varnothing$.
Since the chosen $(h+1)$-sets $E$ cover $[m]$, and $T\ne\varnothing$,
some chosen set $E$ meets $T$. Fix such an $E$. Then
\[
 \dist(u,q_E)
 =|T\triangle E|
 =|T|+|E|-2|T\cap E|
 =|T|+h+1-2|T\cap E|
 \leq|T|+h-1,
\]
whereas the positive prototype $(1^h,1^m)$ has distance $|T|+h$.
For each negative prototype $q_E$ and each $j\in[h]$, its positive
counterpart $(e_j,\mathbf1_{[m]\setminus E})$ is one unit further from
$u$ than $q_E$. A closest negative prototype is therefore strictly closer than every
positive prototype. This gives a representation of size $(h+1)\ell+1$.

If $m<h+1$, choose $J\subseteq[h]$ with $|J|=m-1$. Take
$N=\{(0^h,0^m)\}$. Put the $h$ auxiliary singletons $(e_j,0^m)$,
$j\in[h]$, in $P$, and add $(\mathbf1_J,1^m)$ to $P$. For a point
$(0^h,x)$ on the face $z=0^h$,
\[
 \dist\bigl((0^h,x),(\mathbf1_J,1^m)\bigr)
   =|J|+m-\wt{x}=2m-1-\wt{x},
\]
whereas
\[
 \dist\bigl((0^h,x),0^{h+m}\bigr)=\wt{x}.
\]
Thus the first distance minus the second is
$2m-1-2\wt{x}$, which is negative exactly when $x=1^m$.
Each auxiliary singleton prototype $(e_j,0^m)$ is at distance $1+\wt{x}$
from $(0^h,x)$, and hence is further away than $0^{h+m}$.
Every remaining positive point has $z_j=1$ for some $j\in[h]$. The positive
prototype $(e_j,0^m)$ is then one unit closer to the point than the sole
negative prototype $0^{h+m}$, so the point is classified positively. This
representation has $h+2$ prototypes, as required.
\end{proof}

\begin{corollary}
\label{cor:orand-parity}
For every $m\geq1$,
\[
 \BNN(z\vee\operatorname{AND}_m)
 =2\left\lceil\frac m2\right\rceil+1
 =\begin{cases}
 \BNN(\operatorname{AND}_m)+1=m+2,&m\text{ odd},\\
 \BNN(\operatorname{AND}_m)=m+1,&m\text{ even}.
 \end{cases}
\]
\end{corollary}

\begin{proof}
Apply Theorem~\ref{thm:orand} with $h=1$ and use
$\BNN(\operatorname{AND}_m)=m+1$, which follows from
Theorem~\ref{thm:covering}.
\end{proof}

The case $m=3$ also shows that the sufficient condition in
Theorem~\ref{thm:general-disjunctive-lift} is not necessary. For
$f=\operatorname{AND}_3=TH_3^3$, the only exposed minimal positive support
is $\{1,2,3\}$. The single block $\{1,2,3\}$ satisfies the covering
requirement, so $c(f)=1$, whereas Theorem~\ref{thm:covering} gives
$\BNN(f)=4$. Thus $\BNN(f)\ne\max\{1,c(f)\}+1$. Even so,
Corollary~\ref{cor:orand-parity} gives
$\BNN(z\vee f)=5=\BNN(f)+1$.

\subsection{Representations with one negative prototype}
\label{subsec:one-negative}

For thresholds $TH_t^n$ with $2t-1\leq n$, the covering construction of
Theorem~\ref{thm:covering} uses exactly one negative prototype. We now
characterise the monotone functions for which such a Boolean representation
exists, and determine its minimum size.

Let $f$ be a non-constant monotone Boolean function on $\cube$. For
$t\in[n]$, let $\mathcal A_t(f)$ be the family of supports of the minimal
positive points of $f$ having weight $t$. Equivalently, $\mathcal A_t(f)$
consists of those $t$-subsets $A\subseteq[n]$ for which $\mathbf1_A$ is a
minimal positive point of $f$.

The number $2t-1$ arises naturally from comparison with the zero
prototype. If $T\subseteq C$ and $|T|=t$, then
\[
 \dist(\mathbf1_T,\mathbf1_C)=|C|-t
 \quad\text{and}\quad
 \dist(\mathbf1_T,0^n)=t.
\]
Hence $\mathbf1_T$ is strictly closer to the positive prototype
$\mathbf1_C$ than to the negative prototype $0^n$ if and only if
$|C|\leq2t-1$.

The covering blocks used for threshold functions require only that a block
contain the relevant minimal positive support. For a general monotone
function, we need an additional condition ensuring that the corresponding
prototype does not make any negative $t$-point positive. This leads to the
following refinement.

We call $C\subseteq[n]$ an \emph{admissible rank-$t$ block} if
\[
 |C|=2t-1
 \quad\text{and}\quad
 f(\mathbf1_T)=1\quad\text{for every }T\in\binom Ct.
\]
If $\mathcal A_t(f)=\varnothing$, set $\kappa_t(f)=0$. Suppose now that
$\mathcal A_t(f)\neq\varnothing$. If there is a family of admissible rank-$t$
blocks such that every member of $\mathcal A_t(f)$ is contained in some block,
let $\kappa_t(f)$ be the minimum number of blocks in such a family. Thus
$\kappa_t(f)$ measures the smallest number of admissible rank-$t$ blocks
needed to cover all minimal positive supports of weight $t$.

\begin{theorem}[Characterisation of one-negative representations]
\label{thm:one-negative}
A non-constant monotone Boolean function $f$ admits a Boolean
nearest-neighbour representation with exactly one negative prototype if
and only if, for every $t\in[n]$, the family $\mathcal A_t(f)$ can be
covered by admissible rank-$t$ blocks. When this holds, the minimum size
of such a representation is
\[
 1+\sum_{t=1}^n\kappa_t(f).
\]
\end{theorem}

\begin{proof}
Suppose first that $(P,\{q\})$ is a representation. Part 2 of
Lemma~\ref{lem:boundary} implies that $q$ lies below every maximal
negative point. We claim that $q_i=0$ on every essential coordinate,
where coordinate $i$ is essential if changing it can change the value of
$f$. Choose a positive point $x$ such that flipping coordinate $i$ from
$1$ to $0$ makes it negative, and choose $A\subseteq\supp(x)$ minimal
under inclusion subject to $f(\mathbf1_A)=1$. Then $i\in A$: otherwise
$A$ would remain contained in the support of the point obtained by flipping
coordinate $i$, which would be positive by monotonicity, a contradiction.
Extend $A\setminus\{i\}$ to a maximal negative support $B$. Since $B$ is
negative it cannot contain $A$, and therefore $i\notin B$. By
Lemma~\ref{lem:boundary}, $q\leq\mathbf1_B$, so
$q_i\leq(\mathbf1_B)_i=0$, and hence $q_i=0$.

If coordinate $i$ is inessential and $q_i=1$, its complementation leaves
$f$ unchanged. By Lemma~\ref{lem:nn-symmetries}, we may apply that
complementation to every prototype without changing the represented
function or the number of prototypes of either label. Doing so on all
such coordinates allows us to assume that $q=0^n$.

Let $A$ be a minimal positive support of size $t$, and choose a
positive prototype $p$ closest to $\mathbf1_A$. By Part 1 of
Lemma~\ref{lem:boundary}, $A\subseteq\supp(p)$. Since $0^n$ is the sole
negative prototype, correct classification of the positive point $\mathbf1_A$
gives
\[
 \dist(\mathbf1_A,p)=\wt p-t<t=\dist(\mathbf1_A,0^n).
\]
Hence $\wt p<2t$, and therefore $\wt p\leq2t-1$. The
one-negative-prototype hypothesis forces equality. Indeed, for any
$i\in A$, the point $\mathbf1_{A\setminus\{i\}}$ is negative. Since $0^n$
is the sole negative prototype and $A\subseteq\supp(p)$,
\[
 \dist(\mathbf1_{A\setminus\{i\}},p)=\wt p-(t-1),
 \qquad
 \dist(\mathbf1_{A\setminus\{i\}},0^n)=t-1.
\]
Correct classification therefore requires
\[
 t-1<\wt p-(t-1),
\]
so $\wt p>2t-2$, and hence $\wt p\geq2t-1$. Combining this with the
preceding upper bound gives
\[
 \wt p=2t-1.
\]

For every $T\in\binom{\supp(p)}t$, we have $T\subseteq\supp(p)$ and
$\wt p=2t-1$. Hence
\[
 \dist(\mathbf1_T,p)=\wt p-t=t-1<t=\dist(\mathbf1_T,0^n),
\]
so $\mathbf1_T$ is positive. We have shown that
\[
 |\supp(p)|=2t-1,
\]
and that every $t$-subset $T\subseteq\supp(p)$ satisfies
\[
 f(\mathbf1_T)=1.
\]
Hence $\supp(p)$ satisfies both conditions in the definition of an
admissible rank-$t$ block. It also covers $A$, since
$A\subseteq\supp(p)$.

Since $A\in\mathcal A_t(f)$ was arbitrary, every member of
$\mathcal A_t(f)$ lies in an admissible rank-$t$ block arising from some
positive prototype. Hence $\mathcal A_t(f)$ can be covered by such blocks,
and, for fixed $t$, at least $\kappa_t(f)$ positive prototypes are needed.
Each admissible rank-$t$ block has size $2t-1$, so its corresponding
prototype has weight $2t-1$. Different values of $t$ require different
weights, and hence the required prototype sets are disjoint. Therefore
\[
 |P|\geq\sum_{t=1}^n\kappa_t(f).
\]
Thus every representation with one negative prototype has size at least
$1+\sum_{t=1}^n\kappa_t(f)$.

Conversely, suppose that, for every $t\in[n]$, the family
$\mathcal A_t(f)$ can be covered by admissible rank-$t$ blocks. For each
$t$, choose a minimum such cover, and use the incidence vectors of all
chosen blocks as positive prototypes, with zero as the sole negative
prototype. If $x$ is positive, it contains a minimal positive support $A$
of some size $t$. A chosen block $C$ contains $A$, and hence
\[
 \dist(x,\mathbf1_C)
 =\wt x+|C|-2|\supp(x)\cap C|.
\]
Since $A\subseteq\supp(x)\cap C$, we have
$|\supp(x)\cap C|\geq t$. As $|C|=2t-1$ and
$\dist(x,0^n)=\wt x$, it follows that
\[
 \dist(x,\mathbf1_C)\leq\dist(x,0^n)-1.
\]
Thus $x$ is classified positively.

If $x$ is negative, then for every admissible rank-$t$ block $C$
we have $|\supp(x)\cap C|\leq t-1$. Indeed, if
$|\supp(x)\cap C|\geq t$, choose
$T\subseteq\supp(x)\cap C$ of size $t$. Admissibility gives
$f(\mathbf1_T)=1$, and monotonicity would then force $x$ to be positive,
a contradiction. Consequently, for the positive prototype $\mathbf1_C$,
\[
 \dist(x,\mathbf1_C)
 =\wt x+2t-1-2|\supp(x)\cap C|
 \geq\wt x+1
 =\dist(x,0^n)+1.
\]
Thus $0^n$ is strictly closer than every positive prototype, so $x$ is
classified negatively. Hence the chosen prototypes form a representation
of size
\[
 1+\sum_{t=1}^n\kappa_t(f).
\]
Together with the lower bound, this proves the result.
\end{proof}

When the minimal positive supports of $f$ all have the same size, the
admissibility condition has a simple hypergraph interpretation.

\begin{corollary}\label{cor:one-negative-uniform}
Suppose that the minimal positive supports of $f$ form a non-empty
$t$-uniform family $\calF$. Then $f$ has a representation with one
negative prototype if and only if, for every $A\in\calF$, there is a set
$C\subseteq[n]$ such that
\[
 A\subseteq C,
 \qquad |C|=2t-1,
 \qquad\text{and}\qquad
 \binom Ct\subseteq\calF.
\]
The minimum total size of such a representation is one plus the minimum
number of these sets $C$ needed to cover $\calF$.

The condition above may equivalently be stated in hypergraph language:
every edge of $\calF$ belongs to a complete $t$-uniform hypergraph on
$2t-1$ vertices whose edges are all contained in $\calF$.
\end{corollary}

\begin{proof}
For any $t$-set $T$, we have $f(\mathbf1_T)=1$ if and only if
$T\in\calF$. Indeed, if $\mathbf1_T$ is positive, then $T$ contains a
minimal positive support; that support has size $t$, and hence equals $T$.
Therefore, for a set $C$ with $|C|=2t-1$, the condition that $C$ be an
admissible rank-$t$ block is equivalent to requiring every $t$-subset
$T\subseteq C$ to be positive, which in turn is equivalent to
\[
 \binom Ct\subseteq\calF.
\] Thus the admissible blocks are
exactly the complete $t$-uniform hypergraphs described in the statement.
The result now follows from Theorem~\ref{thm:one-negative}.
\end{proof}

For graphs, admissibility has a particularly concrete interpretation.

\begin{corollary}\label{cor:graph-one-negative}
Let $G$ be a graph on $[n]$ with at least one edge, and define
\[
 f_G(x)=1
 \quad\Longleftrightarrow\quad
 \supp(x)\text{ contains an edge of }G.
\]
Then $f_G$ has a Boolean nearest-neighbour representation with exactly one
negative prototype if and only if every edge of $G$ lies in a triangle.
When this holds, the minimum size of such a representation is one plus the
minimum number of triangles of $G$ whose edge sets cover $E(G)$.
\end{corollary}

\begin{proof}
The minimal positive supports of $f_G$ are precisely the edges of $G$.
For $t=2$, an admissible block has three vertices and all three of its
two-subsets are edges of $G$; it is therefore exactly the vertex set of a
triangle. The result follows from
Corollary~\ref{cor:one-negative-uniform}.
\end{proof}

For example, $f(x_1,x_2)=x_1\land x_2$ has no Boolean representation
with exactly one negative prototype: a rank-$2$ block would need three
elements in $[2]$. The real prototypes $(0,0)$, labelled negative, and
$(3/2,3/2)$, labelled positive, do represent it. The function
$(x_1\land x_2)\lor(x_3\land x_4)$ gives an obstruction due to
admissibility rather than ambient size. Every three-element subset of
$[4]$ contains a negative pair with one element in each of $\{1,2\}$ and
$\{3,4\}$, so no rank-$2$ block is admissible. Equivalently, its two minimal
positive supports are edges of a graph with no triangle, and
Corollary~\ref{cor:graph-one-negative} excludes a one-negative Boolean
representation.

Taking $\calF=\binom{[n]}t$ in
Corollary~\ref{cor:one-negative-uniform} recovers $TH_t^n$. If
$2t-1\leq n$, this gives the one-negative covering construction of
Theorem~\ref{thm:covering}. If $2t-1>n$, no set
$C\subseteq[n]$ of size $2t-1$ exists, and hence there is no admissible
rank-$t$ block and no representation with exactly one negative prototype;
in this case, applying the dual version of Theorem~\ref{thm:covering} gives
a representation with exactly one positive prototype.

More generally, covering the minimal positive supports by sets of size
$2t-1$ is not enough: admissibility requires every $t$-subset of $C$ to
belong to $\calF$. If this fails, Proposition~\ref{prop:shadow-representation}
still gives a Boolean nearest-neighbour representation, but it may require
more than one negative prototype.

\begin{proposition}\label{prop:shadow-representation}
Let $1\leq t\leq n$, and let
$\calF\subseteq\binom{[n]}t$ be a non-empty $t$-uniform family. Write
\[
 \partial^-\calF
 =\{A\setminus\{a\}:A\in\calF,\ a\in A\}
\]
for its lower shadow.  Define the monotone Boolean function $f$ by
\[
 f(x)=1\quad\Longleftrightarrow\quad
 A\subseteq\supp(x)\text{ for some }A\in\calF.
\]
Thus the minimal positive supports of $f$ are precisely the members of
$\calF$.  Then
\[
 P=\calF,
 \qquad
 N=\partial^-\calF
\]
is a Boolean nearest-neighbour representation of $f$, and
\[
 \BNN(f)\leq\lvert\calF\rvert+\lvert\partial^-\calF\rvert.
\]
\end{proposition}

\begin{proof}
Identify cube points with their supports.

Suppose first that $S$ is positive. Then $S$ contains some $A\in\calF$.
Since $|A|=t$,
\[
 \dist(S,A)=|S|-t.
\]
Every negative prototype $B\in\partial^-\calF$ has size $t-1$. Since
$B$ can contain at most $t-1$ elements of $S$, at least $|S|-t+1$
elements of $S$ are absent from $B$. Hence
\[
 \dist(S,B)\geq|S|-t+1.
\]
Therefore $S$ is strictly closer to the positive prototype $A$ than to
every negative prototype, so $S$ is classified positively.

Now suppose that $S$ is negative. For every positive prototype
$A'\in\calF$,
\[
 \dist(S,A')=|S|+t-2|S\cap A'|,
\]
since $|A'|=t$. Choose $A\in\calF$ so that $|S\cap A|$ is as large as
possible. Then
\[
 \dist(S,A)=\min_{A'\in\calF}\dist(S,A'),
\]
so $A$ is a closest positive prototype to $S$.

Because $S$ is negative, it cannot contain $A$, since otherwise it would
contain a minimal positive support and would be positive by monotonicity.
Hence there is some $a\in A\setminus S$. The set
$A\setminus\{a\}$ belongs to $\partial^-\calF$, so it is a negative
prototype. Since $a\notin S$, removing $a$ from $A$ removes exactly one
disagreement with $S$. Thus
\[
 \dist\bigl(S,A\setminus\{a\}\bigr)=\dist(S,A)-1.
\]
Since $\dist(S,A)$ is the minimum distance from $S$ to any positive
prototype, the negative prototype $A\setminus\{a\}$ is strictly closer to
$S$ than every positive prototype. Hence $S$ is classified negatively.
\end{proof}

For the threshold function $TH_t^n$, Proposition~\ref{prop:shadow-representation}
uses the entire labelled boundary: all weight-$t$ positive prototypes and
all weight-$(t-1)$ negative prototypes. Its size is
\[
 \binom nt+\binom n{t-1}.
\]
By contrast, Theorem~\ref{thm:covering} gives the optimal Boolean
nearest-neighbour complexity. Writing
\[
 r=\min\{t,n-t+1\},
\]
the entropy estimate of Corollary~\ref{cor:entropy} yields
\[
 \frac{\binom nt+\binom n{t-1}}{\BNN(TH_t^n)}
 =2^{2r+O(\log_2 n)},
\]
with an implicit constant independent of $t$. Thus the boundary
construction can be exponentially larger than an optimal representation:
up to polynomial factors, it uses $2^{2r}$ times as many prototypes.

\section{Boolean complexity of a uniformly random function}
\label{sec:uniform-random-limit}

The counting arguments of Section~\ref{sec:capacity} give general almost-all
lower bounds, including for Boolean $k$-nearest-neighbour representations,
simultaneously for every $k$. In the ordinary Boolean nearest-neighbour case,
the component constraints of Section~\ref{sec:boolean-structure} allow a more
detailed description in the uniform random model. After a parity twist, the
argument uses the component structure of site percolation on $Q_n$ at retention
probability $1/2$ to show that the ratio $\BNN(F_n)/2^n$ is asymptotically close
either to $1/2$ or to $1$, with explicit limiting probabilities.

The remainder of this section concerns $\BNN$, as defined in the first
paragraph of Section~\ref{sec:definitions}. In particular, nearest prototypes
of the same label may be tied; the separation requirement is only between
opposite labels.

\begin{theorem}[Two-point limit]\label{thm:uniform-random-limit}
For each $n$, let $F_n$ be uniformly distributed among all Boolean functions
on $\{0,1\}^n$. Put
\[
 p=\bigl(1-e^{-1/2}\bigr)^2=0.154818\ldots.
\]
Then, as $n\to\infty$, for every $0<\varepsilon<1/2$,
\[
 \Pr\!\left(
  \left(\frac12-\varepsilon\right)2^n
  <\BNN(F_n)<
  \left(\frac12+\varepsilon\right)2^n
 \right)\longrightarrow 1-p
\]
and
\[
 \Pr\!\left(
  (1-\varepsilon)2^n<\BNN(F_n)\leq 2^n
 \right)\longrightarrow p.
\]
Moreover,
\[
 \lim_{n\to\infty}\frac{\mathbb E\BNN(F_n)}{2^n}
 =1-e^{-1/2}+\frac12e^{-1}=0.577409\ldots.
\]
\end{theorem}

Thus $\BNN(F_n)/2^n$ is close to $1/2$ with limiting probability
$0.845181\ldots$, and close to $1$ with limiting probability
$0.154818\ldots$. In particular, $1/2$ is the sharp constant in an almost-all
lower bound: for every $\varepsilon>0$, the probability that
$\BNN(F_n)\geq(1/2-\varepsilon)2^n$ tends to one, whereas no larger constant
has this property.

For Boolean voting rules with $k\geq2$,
Corollary~\ref{cor:almost-all-BNN} still provides an almost-all lower bound
simultaneously over all such $k$.

Throughout this section, put $M=2^n$ and write $V_n=\{0,1\}^n$ for the
vertex set of $Q_n$. Recall that two vertices of $Q_n$ are adjacent
precisely when they differ in one coordinate, and that
$H_f$ is the spanning subgraph whose edges join vertices at which $f$ takes
different values.

Let $\pi(x)=|x|\pmod2$ and put $g(x)=f(x)\oplus\pi(x)$, where $\oplus$
denotes addition modulo two. Along an edge of $Q_n$, $f$ changes precisely
when $g$ does not. It follows that the connected components of $H_f$ are
exactly the connected components of the two induced graphs
$Q_n[g^{-1}(0)]$ and $Q_n[g^{-1}(1)]$.

\begin{lemma}\label{lem:omit-alternating-component}
Let $C$ be a connected component of $H_f$.  If every vertex of $C$ has a
Hamming neighbour outside $C$, then $V_n\setminus C$, with each prototype
labelled by $f$ at its location, represents $f$.
\end{lemma}
\begin{proof}
Let $x\in C$. By hypothesis, $x$ has a Hamming neighbour $y\notin C$.
Since the proposed prototype set is $V_n\setminus C$, this vertex $y$ is a
prototype. We have $f(x)=f(y)$:
otherwise $xy$ would be an edge of $H_f$, which would put $x$ and $y$ in
the same connected component, contrary to $y\notin C$. Thus $x$ has a
same-label prototype at distance 1. If an opposite-label prototype $z$
were also at distance 1 from $x$, then $xz$ would be an edge of $H_f$,
so $z$ would belong to $C$, contradicting the choice of the prototype set.
Hence every opposite-label prototype is at distance at least 2 from $x$,
and $x$ is classified correctly. Every vertex outside $C$ is itself a
prototype and therefore classifies itself correctly.
\end{proof}

We use the following theorem of Sapozhenko \cite{Saposhenko75} on site percolation
on $Q_n$ at retention probability $1/2$, stated only in the form needed below.

\begin{theorem}[Sapozhenko]\label{thm:percolation-giant}
Let $U$ be obtained by retaining every vertex of $Q_n$ independently with
probability $1/2$, and let $C$ be a largest connected component of $Q_n[U]$
(empty if $U$ is empty). Then, for every $\varepsilon>0$, the probability
that
\[
 \left|\frac{|C|}{2^n}-\frac12\right|<\varepsilon
\]
and that every connected component of $Q_n[U]$ other than $C$ is a single
vertex tends to one as $n\to\infty$.
\end{theorem}

\begin{remark}\label{rem:percolation-provenance}
The proof uses only the component-structure conclusion of this percolation
theorem. It does not require the accompanying limiting law for the number of
isolated vertices. See also Weber~\cite{Weber83}; Bollob\'as, Kohayakawa, and
\L uczak~\cite[Section~1]{BKL94} give an accessible English-language account.
\end{remark}

Write $B_1(x)$ for the closed Hamming ball of radius 1 about $x\in V_n$.

\begin{lemma}\label{lem:monochromatic-balls-poisson}
Let $G_n:V_n\to\{0,1\}$ be random, with the values $G_n(x)$, $x\in V_n$,
chosen independently and each equally likely to be $0$ or $1$. For
$a\in\{0,1\}$, let $Z_a$ be the number, possibly zero, of vertices
$x\in V_n$ for which $G_n(y)=a$ for every $y\in B_1(x)$.
Then, as $n\to\infty$, for every pair of non-negative integers $i,j$,
\[
 \Pr(Z_0=i,Z_1=j)\longrightarrow
 \frac{e^{-1}}{2^{i+j}i!j!}.
\]
In particular,
\[
 \Pr(Z_a=0)\longrightarrow e^{-1/2}\quad(a=0,1),
 \qquad
 \Pr(Z_0=0,Z_1=0)\longrightarrow e^{-1}.
\]
\end{lemma}
\begin{proof}
For a fixed $x\in V_n$ and $a\in\{0,1\}$, the ball $B_1(x)$
contains $n+1$ vertices. Since their values are chosen independently, the
probability that $G_n(y)=a$ for every $y\in B_1(x)$ is $2^{-(n+1)}$.
Summing over the $M$ possible centres gives
\[
 \mathbb E Z_a=M2^{-(n+1)}=\frac12.
\]

For a fixed realisation $g$ of $G_n$, let
\[
 X_a=\{x\in V_n:g(y)=a\text{ for every }y\in B_1(x)\},
\]
so that $Z_a=|X_a|$. Write $(z)_j=z(z-1)\cdots(z-j+1)$ for $j\geq1$,
and set $(z)_0=1$. We will prove that, for every pair of non-negative
integers $r_0,r_1$,
\[
 \mathbb E\bigl[(Z_0)_{r_0}(Z_1)_{r_1}\bigr]
 \longrightarrow (1/2)^{r_0+r_1}.
\]

Fix $r_0,r_1$ and put $r=r_0+r_1$. If $r=0$, then both sides equal $1$.
If $r=1$, the required limit follows from $\mathbb E Z_a=1/2$. We may
therefore suppose that $r\geq2$.

The product $(Z_0)_{r_0}(Z_1)_{r_1}$ counts ordered tuples
\[
 (x_1,\ldots,x_{r_0},y_1,\ldots,y_{r_1}),
\]
where the $x_i$ are distinct members of $X_0$ and the $y_j$ are distinct
members of $X_1$. Since $X_0\cap X_1=\varnothing$, all $r$ centres in the
tuple are distinct. The balls centred at the first $r_0$ entries are
prescribed value $0$, and those centred at the remaining $r_1$ entries are
prescribed value $1$. The expectation is obtained by summing, over all
possible ordered tuples of distinct centres, the probability that their
prescribed ball conditions hold.

Call such a tuple well-separated if every two of its centres have Hamming
distance at least 3, and let $A_n$ be the number of well-separated
tuples. There are
\[
 (M)_r=M(M-1)\cdots(M-r+1)
\]
ordered tuples of $r$ distinct centres altogether.

We now bound the number that are not well-separated. There are $\binom r2$
ways to choose two entries in the tuple that might form a close pair. For
each such pair, there are $M$ choices for the first centre,
\[
 n+\binom n2
\]
choices for the second centre at Hamming distance 1 or 2 from it, and at
most $M^{r-2}$ choices for the remaining centres. Consequently,
\[
 0\leq (M)_r-A_n
 \leq \binom r2 M^{r-1}\left(n+\binom n2\right).
\]
This is an upper bound because a tuple containing more than one close pair
may be counted more than once. Dividing by $M^r$ gives
\[
 0\leq \frac{(M)_r}{M^r}-\frac{A_n}{M^r}
 \leq \binom r2\frac{n+\binom n2}{M}.
\]
Since $r$ is fixed and $M=2^n$,
\[
 \frac{(M)_r}{M^r}
 =\prod_{j=0}^{r-1}\left(1-\frac{j}{M}\right)\longrightarrow1
\]
and
\[
 \binom r2\frac{n+\binom n2}{M}\longrightarrow0.
\]
It follows that
\[
 \frac{A_n}{M^r}\longrightarrow1.
\]

For every well-separated tuple, the $r$ radius-1 balls are disjoint. The
prescribed conditions therefore fix the values of $g$ at $r(n+1)$ distinct
vertices, so the probability that they all hold is $2^{-r(n+1)}$. The contribution of the well-separated tuples to
$\mathbb E[(Z_0)_{r_0}(Z_1)_{r_1}]$ is therefore
\[
 A_n2^{-r(n+1)}
 =\frac{A_n}{M^r}\,2^{-r}
 \longrightarrow2^{-r}.
\]

It remains to consider the tuples that are not well-separated. Each of the
$r$ balls contains $n+1$ vertices, so counting the balls separately gives
$r(n+1)$ vertices. Any two distinct radius-1 balls intersect in at most
two vertices, and there are $\binom r2$ pairs of balls. It follows that their
union contains at least
\[
 r(n+1)-2\binom r2
\]
vertices.

If two prescribed values conflict on an intersection, the probability
associated with the tuple is zero. Otherwise, the ball conditions prescribe
the values of $g$ throughout their union, so the probability is at most
\[
 2^{-r(n+1)+2\binom r2}.
\]
Using the preceding bound on the number of tuples that are not well-separated,
their contribution to the same expectation is at most
\[
 \binom r2 M^{r-1}\left(n+\binom n2\right)
 2^{-r(n+1)+2\binom r2}.
\]
Since $M=2^n$, this is
\[
 2^{2\binom r2-r}\binom r2\frac{n+\binom n2}{M},
\]
which tends to zero. Combining the contributions from the well-separated
and remaining tuples gives
\[
 \mathbb E\bigl[(Z_0)_{r_0}(Z_1)_{r_1}\bigr]
 \longrightarrow 2^{-r}=(1/2)^{r_0+r_1}.
\]

The multivariate factorial-moment criterion
\cite[Theorem~6.10]{JLR00}, applied with $\lambda_0=\lambda_1=1/2$, now
gives
\[
 \Pr(Z_0=i,Z_1=j)\longrightarrow
 \frac{e^{-1}}{2^{i+j}i!j!}.
\]
In particular, each marginal converges in distribution to a Poisson random
variable with mean $1/2$. Hence
\[
 \Pr(Z_a=0)\longrightarrow e^{-1/2}\qquad(a=0,1),
\]
while the joint limit at $(0,0)$ gives
\[
 \Pr(Z_0=0,Z_1=0)\longrightarrow e^{-1}.
\]
\end{proof}

\begin{lemma}[Deterministic component bounds]
\label{lem:deterministic-component-bounds}
Let $n\geq1$, let $f:V_n\to\{0,1\}$, and put $g=f\oplus\pi$.
For $a\in\{0,1\}$, let $C_a$ be a largest component of
$Q_n[g^{-1}(a)]$, taking $C_a=\varnothing$ if $g^{-1}(a)$ is empty, and put
\[
 R=V_n\setminus(C_0\cup C_1).
\]
Let $Z_a$ count the vertices $x$ for which $g=a$ throughout $B_1(x)$.
Then the following statements hold.
\begin{enumerate}
 \item If $\BNN(f)>|R|$, then $\BNN(f)\geq\min\{|C_0|,|C_1|\}$.
 \item Suppose that every component of $Q_n[g^{-1}(0)]$ or
 $Q_n[g^{-1}(1)]$ other than $C_0,C_1$ is a single vertex. If $Z_a>0$,
 every representation contains all of $C_a$, so $\BNN(f)\geq|C_a|$.
 In particular, if $Z_0>0$ and $Z_1>0$, then
 $\BNN(f)\geq|C_0|+|C_1|$.
 \item If $Z_a=0$, then $\BNN(f)\leq M-|C_a|$.
\end{enumerate}
\end{lemma}

\begin{proof}
For Part~1, if either $C_a$ is empty, the conclusion is immediate.
Otherwise $C_0,C_1$ are components of $H_f$, so
Proposition~\ref{prop:edge-closure} makes each of them either wholly
present or wholly absent in a representation. A representation containing
neither would have all its prototypes in $R$ and would therefore have size
at most $|R|$, contrary to $\BNN(f)>|R|$. This proves Part~1.

Under the hypothesis of Part~2, a vertex counted by $Z_a$ has all its
neighbours in $g^{-1}(a)$, so it is not isolated and therefore belongs to
$C_a$. All its neighbours have the opposite $f$-value.
Lemma~\ref{lem:monochromatic-shortest-paths} makes it a compulsory
prototype, and Proposition~\ref{prop:edge-closure} then forces all of
$C_a$. If both counts are positive, the two disjoint components are
compulsory, proving Part~2.

If $Z_a=0$, every vertex of $C_a$ has a cube neighbour outside $C_a$:
otherwise all vertices in its closed radius-1 ball would have $g$-value $a$.
When $C_a$ is non-empty, Lemma~\ref{lem:omit-alternating-component} gives
a representation with prototype set $V_n\setminus C_a$, proving Part~3.
If $C_a$ is empty, the same bound follows by taking all cube points as
prototypes.
\end{proof}

\begin{proof}[Proof of Theorem~\ref{thm:uniform-random-limit}]
Adding parity turns bichromatic components into monochromatic components.
Theorem~\ref{thm:percolation-giant} supplies the component structure and
size estimates needed to apply
Lemma~\ref{lem:deterministic-component-bounds}; the counts $Z_0,Z_1$
then determine the two alternatives.

Put $g=F_n\oplus\pi$. Since parity is deterministic, the values $g(x)$
are again independent and each is equally likely to be $0$ or $1$.
For each $a\in\{0,1\}$, let $C_a$ be a largest component of
$Q_n[g^{-1}(a)]$, taking $C_a=\varnothing$ if $g^{-1}(a)$ is empty, and put
$R=V_n\setminus(C_0\cup C_1)$. Define
\[
 A_n=\{Z_0>0,\ Z_1>0\},\qquad
 X_n=\frac{\BNN(F_n)}M,\qquad
 Y_n=\frac12+\frac12\mathbf1_{A_n}.
\]

Fix $0<\eta\leq1/10$. Apply
Theorem~\ref{thm:percolation-giant} to the two induced subgraphs
$Q_n[g^{-1}(0)]$ and $Q_n[g^{-1}(1)]$ and take a union bound. With
probability tending to one, every component of $Q_n[g^{-1}(a)]$ other than
$C_a$ is a single vertex for each $a\in\{0,1\}$, and
\[
 \left(\frac12-\frac\eta2\right)M<|C_a|<
 \left(\frac12+\frac\eta2\right)M\quad(a=0,1).
\]
Since $C_0$ and $C_1$ are disjoint, this gives
\[
 |C_0|+|C_1|>(1-\eta)M,
\]
and hence
\[
 |R|=M-|C_0|-|C_1|<\eta M.
\]
Corollary~\ref{cor:almost-all-BNN}, applied with $Q=M$ and $c=1/5$, gives
\[
 \Pr\bigl(\BNN(F_n)\leq M/5\bigr)
 \leq \lfloor M/5\rfloor\,2^{-M\delta(1/5)}
 \longrightarrow0,
\]
since $H_2(1/5)+1/5<1$. Let $\mathcal G_n(\eta)$ be the event that the
displayed bounds hold, that every component of $Q_n[g^{-1}(a)]$ other than
$C_a$ is a single vertex for each $a\in\{0,1\}$, and that $\BNN(F_n)>M/5$.
A union bound gives $\Pr(\mathcal G_n(\eta))\to1$.

Recall that $A_n=\{Z_0>0,\ Z_1>0\}$. On $\mathcal G_n(\eta)$, we have
$\BNN(F_n)>|R|$, so Lemma~\ref{lem:deterministic-component-bounds} applies.
If $A_n$ occurs,
Part~2 and the all-vertices representation give
\[
 1-\eta<X_n\leq1.
\]
If $A_n$ does not occur, Parts~1 and~3 give
\[
 \frac12-\frac\eta2<X_n<\frac12+\frac\eta2.
\]
Thus $|X_n-Y_n|<\eta$ on $\mathcal G_n(\eta)$. Choosing
$0<\eta<\min\{\delta,1/10\}$ proves that, for every $\delta>0$,
\begin{equation}\label{eq:uniform-two-point-comparison}
 \Pr(|X_n-Y_n|>\delta)\longrightarrow0.
\end{equation}
Equivalently,
\[
 X_n-Y_n\xrightarrow{\Pr}0,
\]
where $\xrightarrow{\Pr}$ denotes convergence in probability.

By Lemma~\ref{lem:monochromatic-balls-poisson} and inclusion--exclusion,
\[
\begin{aligned}
 \Pr(A_n)
 &=1-\Pr(Z_0=0)-\Pr(Z_1=0)+\Pr(Z_0=0,Z_1=0)\\
 &\longrightarrow1-2e^{-1/2}+e^{-1}
 =(1-e^{-1/2})^2=p.
\end{aligned}
\]
It remains to translate the comparison with $Y_n$ into the two intervals in
the theorem. Fix $0<\varepsilon<1/2$, and choose
\[
 \delta=\frac12\min\left\{\varepsilon,\frac12-\varepsilon\right\}.
\]
Then $\delta<\varepsilon$ and $\delta<1/2-\varepsilon$. Hence, on the event
$\{|X_n-Y_n|\leq\delta\}$, if $A_n$ occurs then $Y_n=1$, and so
$X_n>1-\varepsilon$; while if $A_n$ does not occur then $Y_n=1/2$, and so
\[
 \frac12-\varepsilon<X_n<\frac12+\varepsilon.
\]
Since $\delta<1/2-\varepsilon$, when $A_n$ occurs we also have
$X_n\geq1-\delta>1/2+\varepsilon$, while when $A_n$ fails we have
$X_n\leq1/2+\delta<1-\varepsilon$. Thus, on
$\{|X_n-Y_n|\leq\delta\}$, the interval near $1/2$ in the theorem occurs
precisely when $A_n$ fails, and the interval near $1$ occurs precisely when
$A_n$ occurs.
The probability of the complementary event $\{|X_n-Y_n|>\delta\}$ tends to
zero by \eqref{eq:uniform-two-point-comparison}, proving the two probability
limits.

Finally, since $0\leq X_n,Y_n\leq1$, for every $\delta>0$,
\[
 \left|\mathbb E X_n-\frac12-\frac12\Pr(A_n)\right|
 \leq\mathbb E|X_n-Y_n|
 \leq\delta+\Pr(|X_n-Y_n|>\delta).
\]
Letting $n\to\infty$ and then $\delta\downarrow0$ gives
\[
 \mathbb E X_n\longrightarrow\frac12+\frac p2
 =1-e^{-1/2}+\frac12e^{-1}.
\]
\end{proof}

\section{Symmetric functions: approximation and typical complexity}
\label{sec:random-symmetric}

For a uniformly random Boolean function, Section~\ref{sec:uniform-random-limit}
shows that $\BNN(F_n)/2^n$ is asymptotically close either to $1/2$ or to
$1$. For symmetric functions, the same component constraints lead to a
deterministic approximation by weighted vertex covers on paths.

We first establish this approximation for every symmetric function. We
then apply it to independent changes of layer label, obtaining
concentration at $11/20$ in the uniform symmetric model and a formula for
arbitrary fixed change probability.

\subsection{A deterministic approximation for Boolean prototypes}

By Lemma~\ref{lem:monochromatic-shortest-paths}, every monochromatic component
contains a prototype. A layer whose label differs from both adjacent
layers consists entirely of compulsory prototypes, since each point in
that layer is an isolated monochromatic component; see also \cite{CG25}
for the middle-layer case.

Figure~\ref{fig:symmetric-layer-structure} illustrates the definitions below
for the function considered in Example~\ref{ex:symmetric-eight}.

Let $f$ be a symmetric Boolean function, and write $\xi_i$ for its
common value on the $i$-th Hamming layer, $0\leq i\leq n$. Thus
$f(x)=\xi_{|x|}$. For $1\leq i\leq n$, define
\[
 D_i=\mathbf1_{\{\xi_{i-1}\ne\xi_i\}},
\]
and, for bookkeeping, put $D_0=D_{n+1}=1$. A \emph{transition} is an
index $i\in\{1,\ldots,n\}$ with $D_i=1$, indicating a change between
layers $i-1$ and $i$. A \emph{constant run} is a maximal interval
$[u,v]\subseteq\{0,\ldots,n\}$ such that layers $u,\ldots,v$ have the same
label. Layers $u$ and $v$ are its boundary layers, and any layers strictly
between them are its interior layers. Two constant runs separated by a
transition are called neighbouring runs. Consecutive transitions force
complete layers of prototypes. Two-layer constant runs create a different
either--or constraint, which is recorded by the graph introduced below. To
distinguish these constraints,
partition the transitions into
\[
 \mathcal T=\{i\in[n]:D_i=1,\ D_{i-1}+D_{i+1}\geq1\},
 \qquad
 \mathcal I=\{i\in[n]:D_i=1,\ D_{i-1}=D_{i+1}=0\}.
\]
The endpoint convention places any transition at $1$ or $n$ in
$\mathcal T$. For $2\leq i\leq n-1$, a transition $i$ belongs to
$\mathcal T$ when
another transition occurs immediately before or immediately after it. Thus
two consecutive transitions enclose a layer whose label differs from the
labels of the layers on both sides. The transitions in $\mathcal I$ are
called isolated; an isolated transition is not immediately adjacent to
another change. Define
\[
 \mathcal H=\bigcup_{i\in\mathcal T}\{i-1,i\},
 \qquad h_n(f)=\sum_{w\in\mathcal H}\binom nw.
\]
The graph defined next records the choices created by two-layer constant
runs. If $i$ and $i+2$ are isolated transitions, then layers $i$ and $i+1$
form a two-layer constant run. A representation must then choose one of the
two adjacent transition pairs: either layers $i-1,i$, or layers $i+1,i+2$,
are forced. We encode this either--or requirement by an edge joining $i$ and
$i+2$.

Let $G_f$ be the graph with vertex set $\mathcal I$ in which two indices are
adjacent precisely when their difference is $2$. Its components are paths,
since edges only join indices differing by $2$ within the odd and even index
chains. A
vertex cover of $G_f$ is a set $S\subseteq\mathcal I$ containing at least one
of the two vertices of every edge. Give vertex $i$ the weight
\[
 t_i=\binom n{i-1}+\binom ni,
\]
and define
\[
 \tau_{\mathrm{wt}}(G_f)
 =
 \min\left\{
   \sum_{i\in S}t_i:
   S\text{ is a vertex cover of }G_f
 \right\}.
\]
For an edgeless graph this minimum is zero. An edge joining $i$ and $i+2$
may be covered at cost $t_i$ by selecting layers $i-1,i$, or at cost
$t_{i+2}$ by selecting layers $i+1,i+2$. Set
\begin{equation}\label{eq:symmetric-W}
 W_n(f)=h_n(f)+\tau_{\mathrm{wt}}(G_f).
\end{equation}

\begin{figure}[htbp]
\centering
\begin{tikzpicture}[x=10mm,y=9mm,font=\small,line width=1pt]
  \node[anchor=east] at (-0.7,3.2) {$G_f$};
  \draw (1.5,3.2) -- (3.5,3.2);
  \draw[dashed] (1.5,2.83) -- (1.5,2.25);
  \draw[dashed] (3.5,2.83) -- (3.5,2.25);
  \node[circle,draw,fill=black!18,minimum size=6mm,inner sep=0pt]
    at (1.5,3.2) {$2$};
  \node[circle,draw,fill=white,minimum size=6mm,inner sep=0pt]
    at (3.5,3.2) {$4$};
  \node at (1.5,3.92) {$t_2=36$};
  \node at (3.5,3.92) {$t_4=126$};
  \node[anchor=west,align=left] at (5.15,3.2)
    {$\mathcal I=\{2,4\}$\\[2pt]$\mathcal T=\{7,8\}$};
  \node[anchor=east] at (-0.7,1.92) {boundary $i$};
  \node[anchor=east] at (-0.7,1.02) {change $D_i$};
  \foreach \i/\bit in {1/0,2/1,3/0,4/1,5/0,6/0,7/1,8/1} {
    \node at (\i-0.5,1.92) {$\i$};
    \node at (\i-0.5,1.02) {$\bit$};
  }
  \foreach \i in {2,4,7,8} {
    \draw[dashed] (\i-0.5,0.65) -- (\i-0.5,-0.36);
  }
  \node[anchor=east] at (-0.7,0.1) {label $\xi_w$};
  \foreach \w/\val in {0/0,1/0,2/1,3/1,4/0,5/0,6/0,7/1,8/0} {
    \draw[fill=white] (\w-0.35,-0.28) rectangle (\w+0.35,0.48);
  }
  \foreach \w in {1,2} {
    \draw[fill=black!18] (\w-0.35,-0.28) rectangle (\w+0.35,0.48);
  }
  \foreach \w in {6,7,8} {
    \draw[line width=1.8pt] (\w-0.35,-0.28) rectangle (\w+0.35,0.48);
  }
  \foreach \w/\val in {0/0,1/0,2/1,3/1,4/0,5/0,6/0,7/1,8/0} {
    \node at (\w,0.1) {$\val$};
    \node at (\w,-0.72) {$\w$};
  }
  \node[anchor=east] at (-0.7,-0.72) {layer $w$};
  \foreach \u/\v in {0/1,2/3,4/6,7/7,8/8} {
    \draw[decorate,decoration={brace,mirror,amplitude=3pt}]
      (\u-0.35,-1.09) -- (\v+0.35,-1.09);
  }
  \node[anchor=east] at (-0.7,-1.2) {constant runs};
  \draw[fill=black!18] (0,-2.08) rectangle (0.35,-1.7);
  \node[anchor=west] at (0.5,-1.89) {chosen by the cover};
  \draw[line width=1.8pt] (4.6,-2.08) rectangle (4.95,-1.7);
  \node[anchor=west] at (5.1,-1.89) {compulsory layer};
\end{tikzpicture}
\caption{Layer structure for the function in Example~\ref{ex:symmetric-eight}.
Boundary $i$ lies between layers $i-1$ and $i$; it is a transition when
$D_i=1$. The bookkeeping values $D_0=D_9=1$ are not shown.
The isolated transitions $2$ and $4$ form the single edge of $G_f$.
A minimum-weight vertex cover selects $2$ (shaded), choosing layers $1,2$.
The consecutive transitions $7,8$ force the complete layers
$\mathcal H=\{6,7,8\}$ (heavy outlines). Braces mark the constant runs.
These complete layers give the lower bound $W_8(f)=37+36=73$;
the example shows why another $20$ prototypes are needed.}
\label{fig:symmetric-layer-structure}
\end{figure}
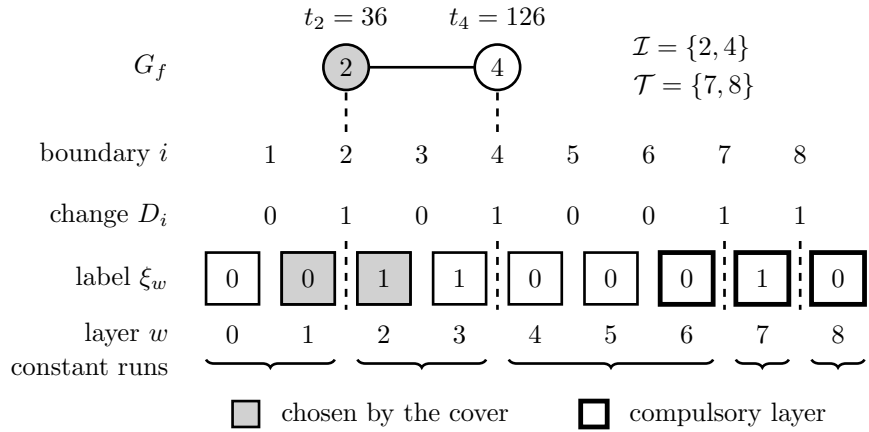

\begin{lemma}\label{lem:symmetric-full-layers}
In every Boolean nearest-neighbour representation of $f$, the following
statements hold.
\begin{enumerate}
 \item If $i$ is a transition and some point in layer $i-1$ or layer $i$ is
 a prototype, then every point in both layers $i-1$ and $i$ is a prototype.
 \item For every $i\in\mathcal T$, every point in layers $i-1$ and $i$ is a
 prototype.
 \item For every edge $\{i,i+2\}$ of $G_f$, either every point in layers
 $i-1$ and $i$ is a prototype, or every point in layers $i+1$ and $i+2$ is
 a prototype. Equivalently, the set of $j\in\mathcal I$ for which every point
 in layers $j-1$ and $j$ is a prototype is a vertex cover of $G_f$.
\end{enumerate}
Moreover, $\BNN(f)\geq W_n(f)$.
\end{lemma}

\begin{proof}
Identify each cube point with its support. The cube edges between layers
$i-1$ and $i$ join an $(i-1)$-subset $A$ to an $i$-subset $B$ precisely when
$A\subset B$. The graph formed by these two layers and the edges between
them is connected. Indeed, two $(i-1)$-subsets can be joined by a sequence
of single-element exchanges, with each exchange passing through an
$i$-subset, and every $i$-subset is adjacent to an $(i-1)$-subset. For
$i=1$, every singleton is joined to the empty set; for $i=n$, every
$(n-1)$-subset is joined to $[n]$.

If $i$ is a transition and some point in layer $i-1$ or layer $i$ is a
prototype, then all edges between these two layers are bichromatic. Since the
graph formed by the two layers is connected, Proposition~\ref{prop:edge-closure}
propagates prototype membership throughout both layers. Hence every point in
both layers is a prototype. This proves Part~1.

Suppose that $D_i=D_{i+1}=1$ for some $1\leq i\leq n-1$. The label of layer
$i$ then differs from the labels of layers $i-1$ and $i+1$. Every point in
layer $i$ therefore has only oppositely labelled neighbours and must be a
prototype by Lemma~\ref{lem:monochromatic-shortest-paths}. Applying Part~1 to the
transitions $i$ and $i+1$ shows that every point in layers $i-1$, $i$, and
$i+1$ is a prototype. Equivalently, applying Part~1 to both transition
indices gives the three layers $i-1,i,i+1$. The case $D_{i-1}=D_i=1$ is
the same argument with the index shifted by one.

The first and last layers require the same conclusion. If $D_1=1$, every
neighbour of $0^n$ lies in layer $1$ and has the opposite label.
Lemma~\ref{lem:monochromatic-shortest-paths} therefore shows that $0^n$ must be a
prototype, and Part~1 then forces every point in layer $1$ to be a prototype.
Similarly, if $D_n=1$, then $1^n$ must be a prototype and Part~1 forces every
point in layer $n-1$ to be a prototype.

This proves Part~2, including the cases represented by the bookkeeping values
$D_0=D_{n+1}=1$.

Now consider an edge of $G_f$ joining $i$ and $i+2$. Layers $i$ and $i+1$
have the same label, while each of layers $i-1$ and $i+2$ has the opposite
label. Layers $i$ and $i+1$ together form a monochromatic connected component
and therefore contain a prototype by
Lemma~\ref{lem:monochromatic-shortest-paths}. If such a prototype lies in
layer $i$, Part~1 at transition $i$ forces layers $i-1,i$; if it
lies in layer $i+1$, Part~1 at transition $i+2$ forces layers $i+1,i+2$.
We represent this
either--or requirement by an edge joining $i$ and $i+2$. This proves Part~3.

Let
\[
 S=\{i\in\mathcal I:
       \text{every point in layers $i-1$ and $i$ is a prototype}\}.
\]
Part~3 shows that $S$ contains at least one of the two vertices of every edge
of $G_f$, so $S$ is a vertex cover. Hence
\[
 \sum_{i\in S}t_i\geq\tau_{\mathrm{wt}}(G_f).
\]
Since isolated transitions are not consecutive, the pairs of layers
$i-1,i$, for $i\in S$, are mutually disjoint. They are also disjoint from
the layers whose indices belong to $\mathcal H$. Part~2 therefore shows that
the number of prototypes in the representation is at least
\[
 h_n(f)+\sum_{i\in S}t_i
 \geq h_n(f)+\tau_{\mathrm{wt}}(G_f)
 =W_n(f).
\]
Since the representation was arbitrary, $\BNN(f)\geq W_n(f)$.
\end{proof}

\begin{example}[A covering gap on eight variables]
\label{ex:symmetric-eight}
Let $f$ have layer-label sequence
\[
 (\xi_0,\ldots,\xi_8)=(0,0,1,1,0,0,0,1,0).
\]
Its change sequence is $(D_1,\ldots,D_8)=(0,1,0,1,0,0,1,1)$, and
\[
 \mathcal T=\{7,8\},\qquad
 \mathcal I=\{2,4\},\qquad
 \mathcal H=\{6,7,8\}.
\]
The graph $G_f$ is the single edge $\{2,4\}$. Its weights and the
compulsory-layer contribution are
\[
 t_2=\binom81+\binom82=36,\qquad
 t_4=\binom83+\binom84=126,\qquad
 h_8(f)=\binom86+\binom87+\binom88=37.
\]
Thus $W_8(f)=37+36=73$: a minimum-weight cover chooses transition $2$
and hence layers $1,2$, in addition to the compulsory layers $6,7,8$.
The alternative choice of transition $4$ gives total layer cost
$37+126=163$.

In fact, $\BNN(f)$ turns out to be $93$. The extra contribution beyond
$W_8(f)=73$ comes from a covering requirement in layer $5$. The complete
layers $1,2,6,7,8$ almost give a representation: points in the unselected
layers $0,3,5$ have same-labelled prototypes at distance $1$ in the selected
layers $1,2,6$, respectively. The remaining problem is layer $4$. A point in
layer $4$ has label $0$, and its same-labelled neighbours lie in layer $5$,
which has not been taken wholesale. Thus we need enough selected layer-5
points so that every layer-4 point has a same-labelled prototype at distance
$1$. Identifying layer-4 points with $4$-subsets of $[8]$ and layer-5 points
with $5$-subsets, this is precisely an $(8,5,4)$-covering requirement.

Using the established value $\Cover(8,5,4)=20$ \cite{GS07}, choose a
minimum $(8,5,4)$-covering. For the upper bound, take all points in layers
$1,2,6,7,8$, together with the twenty layer-5 points corresponding to the
blocks of this covering. Points in layers $0$, $3$, and $5$ have
same-labelled prototypes at distance $1$ in layers $1$, $2$, and $6$,
respectively. If $x$ lies in layer $4$, the covering property gives a
selected layer-5 superset of $x$, again at distance $1$ and with the same
label. From each point considered, every oppositely labelled prototype is
at distance at least $2$, and every selected point classifies itself.
Thus these $73+20=93$ points represent $f$.

For the lower bound, let $S$ be an optimal Boolean prototype set. The
construction above shows that $|S|\leq93$. By
Lemma~\ref{lem:symmetric-full-layers}, $S$ contains all points in layers
$6,7,8$, and, since $G_f$ has the edge $\{2,4\}$, it also contains either
all points in layers $1,2$ or all points in layers $3,4$. The latter
alternative would already contribute
\[
 \binom83+\binom84+\binom86+\binom87+\binom88=163
\]
prototypes, contradicting $|S|\leq93$. Hence $S$ contains all points in
layers $1,2,6,7,8$, giving $73$ prototypes. Moreover, no point in layer $3$
or $4$ is a prototype: by Part~1 of Lemma~\ref{lem:symmetric-full-layers},
any such prototype would force the complete layers $3$ and $4$, contradicting
$|S|\leq93$.

As in the upper-bound construction, these complete layers already account for
points in layers $0,3,5$: they have same-labelled prototypes at distance $1$
in layers $1,2,6$, respectively. The remaining issue is layer $4$. Every
layer-4 point has a positive layer-2 subset at distance $2$, so it needs a
negative prototype at distance $1$. Such a prototype must be a layer-5
superset. The selected layer-5 supports therefore cover every $4$-subset of
$[8]$, requiring at least $\Cover(8,5,4)=20$ additional points. Hence
\[
 \BNN(f)=73+\Cover(8,5,4)=93.
\]
\end{example}

The example shows that $W_n(f)$ need not equal $\BNN(f)$ exactly: after the
forced complete layers have been chosen, some residual partial-layer covering
may still be needed. The next result shows that this residual cost is
negligible compared with $2^n$. Thus the weighted vertex-cover problem on
$G_f$ gives the asymptotic Boolean complexity of every symmetric function.

\begin{theorem}[Weighted vertex-cover approximation]
\label{thm:symmetric-path-approximation}
There is an absolute constant $K$ such that, for every $n\geq2$ and every
symmetric Boolean function $f$ on $\{0,1\}^n$,
\begin{equation}\label{eq:symmetric-path-approximation}
 W_n(f)\leq\BNN(f)
 \leq W_n(f)+K\frac{2^n\log_2 n}{n}.
\end{equation}
In particular, $W_n(f)$ determines $\BNN(f)$ up to an additive $o(2^n)$
error, uniformly over symmetric $f$.
\end{theorem}

\begin{proof}
Lemma~\ref{lem:symmetric-full-layers} gives the lower bound. We prove the
upper bound first for sufficiently large $n$. The construction has four
steps. We select a set $B$ of base prototypes, formed from complete layers,
with $|B|\leq W_n(f)+o(2^n/n)$. We then supply each uncovered point with
at least $\lfloor n/4\rfloor$ same-labelled candidate prototypes nearby,
verify separation from all opposite labels, and use a set-covering argument
to choose a small collection of candidates.

We call a set $L$ of layer indices \emph{closed across transitions} if,
whenever $D_i=1$ and one of $i-1$ and $i$ belongs to $L$, both indices
belong to $L$.

\medskip\noindent
\emph{Claim 1.} There is a set $B$ of labelled cube points, formed from
complete layers whose indices are closed across transitions, such that
\[
 |B|\leq W_n(f)+o(2^n/n).
\]

Choose a minimum-weight vertex cover $S$ of $G_f$, and put
\[
 \mathcal L_0
 =\mathcal H\cup\bigcup_{i\in S}\{i-1,i\}.
\]
The complete layers indexed by $\mathcal L_0$ contain exactly $W_n(f)$
points. Put $a=\lfloor n/4\rfloor$. We include the tail layers wholesale
because their total size is negligible on the $2^n/n$ scale. This will later
ensure that every point still requiring an additional prototype lies in a
central layer, where linearly many nearby same-labelled candidates are
available. Let
$L_B\subseteq\{0,\ldots,n\}$ be
the smallest set containing $\mathcal L_0$ and all indices $w<a$ or
$w>n-a$ that is closed across transitions.
Define
\[
 B=\{x\in\{0,1\}^n:\wt{x}\in L_B\},
\]
with every point carrying its $f$-label.

The set $\mathcal L_0$ is already closed across transitions. Every layer in
$\mathcal L_0$ was selected as one of the two layers adjacent to some
transition $i\in\mathcal T\cup S$. Suppose that this layer is also adjacent
to a distinct transition $j$. Then $|i-j|=1$. This is impossible if
$i\in S$, since $S\subseteq\mathcal I$ and transitions in $\mathcal I$ are
isolated. If $i\in\mathcal T$, then $j\in\mathcal T$ as well, so both
layers adjacent to $j$ already belong to $\mathcal H$.

At the lower end, closing the selected tail layers across transitions can
add only layer $a$, through the transition between layers $a-1$ and $a$.
If the selection could then extend to layer $a+1$, the transitions $a$ and
$a+1$ would both occur. They would therefore belong to $\mathcal T$, and
layer $a+1$ would already be included through $\mathcal H$. Thus the
closure does not extend any further. The same argument at the upper end
shows that it can add only layer $n-a$. Hence
\[
 |B|\leq W_n(f)+T_n,
 \qquad
 T_n=\sum_{w\leq a\ \text{or}\ w\geq n-a}\binom nw.
\]
Since the binomial coefficients increase up to the middle layer,
\[
 T_n\leq2(a+1)\binom na
 =2^{H_2(1/4)n+O(\log_2 n)}.
\]
Put $\delta=1-H_2(1/4)>0$. Then
\[
 \frac{T_n}{2^n/n}
 \leq n2^{-\delta n+O(\log_2 n)}
 \longrightarrow0,
\]
since the $O(\log_2 n)$ term contributes only a polynomial factor. Hence
\begin{equation}\label{eq:symmetric-tail-cost}
 T_n=o(2^n/n).
\end{equation}
This proves Claim~1.

Let $V_{\mathrm{int}}$ be the set of points in interior layers whose indices
do not belong to $L_B$. We say that $x$ is \emph{uncovered by $B$} if
$x\notin B$ and no same-labelled point of $B$ is at distance 1 from $x$.

\medskip\noindent
\emph{Claim 2.} Every point $x$ uncovered by $B$ has a candidate set
$C_x\subseteq V_{\mathrm{int}}$ of size at least $a$, consisting of points
with the same label as $x$, at distance at most one when $x$ lies in a
boundary layer and at most two when it lies in an interior layer.

First note that every point uncovered by $B$ has weight in
$[a+1,n-a-1]$. At weight $a$, layer $a-1$ belongs to $L_B$. If layers
$a-1$ and $a$ have the
same label, it gives every point in layer $a$ a same-labelled point of $B$
at distance 1; if their labels differ, closure across transitions puts
layer $a$ in $L_B$. The argument at weight $n-a$ is symmetric, and all
layers outside $[a,n-a]$ belong to $L_B$.

For a point $x$ uncovered by $B$ and of weight $w$, put
\[
 r_x=
 \begin{cases}
  1,&\text{if layer $w$ is a boundary layer of its run},\\
  2,&\text{if layer $w$ is interior}.
 \end{cases}
\]
For each uncovered point $x$, define the candidate set
\[
 C_x=\{y\in V_{\mathrm{int}}:d_H(x,y)\leq r_x\}.
\]
We first show that every point of $C_x$ has the same label as $x$.

Recall that points of $V_{\mathrm{int}}$ lie only in interior layers of
constant runs. If layer $w$ is a boundary layer of its run and
$y\in V_{\mathrm{int}}$ has $d_H(x,y)\leq1$, then $y$ lies in layer $w-1$
or $w+1$. The adjacent layer across the transition is a boundary layer of
the neighbouring run, so it contains no point of $V_{\mathrm{int}}$. Hence
$y$ lies on the inward side of the same run as $x$, and $f(y)=f(x)$.

Suppose instead that layer $w$ is interior to its run, and let
$y\in V_{\mathrm{int}}$ satisfy $d_H(x,y)\leq2$. Then
$\lvert\wt{y}-w\rvert\leq2$. The layers $w-1,w,w+1$ lie in the same
constant run as $w$, and so have label $f(x)$. If layer $w-2$ or $w+2$ has
the opposite label, then it is a boundary layer of a neighbouring run, and
hence contains no point of $V_{\mathrm{int}}$. Thus $y$ cannot lie in an
oppositely labelled layer, and $f(y)=f(x)$.

We now show that $C_x$ is large. Suppose that $x$ is uncovered and that
layer $w$ is a boundary layer of its run. The run must have at least three
layers, as we now show. A one-layer run in the central range lies between two consecutive
transitions, so its sole layer belongs to $\mathcal H\subseteq L_B$. Now
suppose that the run containing $x$ consists of layers $u$ and $u+1$.
Since $a+1\leq w\leq n-a-1$ and $a\geq1$, both bordering transitions
$u$ and $u+2$ lie in $[n]$. If either transition belongs to
$\mathcal T$, then the adjacent layer of the run belongs to
$\mathcal H\subseteq L_B$. Otherwise $u$ and $u+2$ form an edge of $G_f$,
so the vertex cover $S$ contains at least one of them. If $u\in S$, then
layers $u-1$ and $u$ enter $L_B$; if $u+2\in S$, then layers $u+1$ and
$u+2$ enter $L_B$. Thus in either case one of the two layers in the run
belongs to $L_B$. Every point in a run of one or two layers therefore
either belongs to $B$ or has a same-labelled point of $B$ at distance 1,
and hence cannot be uncovered.

If layer $w$ is the first layer of its run, $C_x$ contains the $n-w$
neighbours of $x$ in layer $w+1$; if it is the last layer, $C_x$ contains
the $w$ neighbours in layer $w-1$. The adjacent layer is interior and does
not belong to $L_B$, since $x$ is uncovered by $B$. So we have
\[
 |C_x|\geq\min\{w,n-w\}\geq a+1.
\]
If layer $w$ is interior, we count only a subset of $C_x$ that is always
present: the point $x$ itself and all points of weight $w$ at distance 2
from it. Points in the neighbouring layers may also belong to $C_x$, but
they are not needed for the lower bound. So we have
\[
 |C_x|\geq1+w(n-w)\geq a.
\]

This proves Claim~2.

\medskip\noindent
\emph{Claim 3.} Let $A\subseteq V_{\mathrm{int}}$ be arbitrary, with every
point of $B\cup A$ carrying its $f$-label. If $x\notin B$ and
$y\in B\cup A$ have opposite labels, then $d_H(x,y)\geq2$. If, in
addition, $x$ is uncovered by $B$ and lies in an interior layer, then
$d_H(x,y)\geq3$.

Fix $x\notin B$, and suppose that $y\in B\cup A$ is at distance 1
from $x$. We claim that $y$ has the same label as $x$.

Suppose otherwise. Since adjacent cube points lie in consecutive layers,
there is a transition between the layer of $x$ and the layer of $y$. In
particular, the layer of $y$ is a boundary layer of its run. Therefore
$y\notin A$, because $A\subseteq V_{\mathrm{int}}$. On the other hand, if
$y\in B$, then the layer of $y$ belongs to $L_B$. Since $L_B$ is closed
across transitions, the layer of $x$ also belongs to $L_B$, which would
imply $x\in B$. Both alternatives are impossible. Hence $y$ has the same
label as $x$.

Now suppose that $x$ is uncovered by $B$ and lies in an interior layer.
Every point of $B\cup A$ at distance 2 from $x$ also has the same label
as $x$. Indeed, an oppositely labelled such point has weight $w-2$ or
$w+2$ and lies in a boundary layer of the neighbouring run, so it cannot
belong to $A$. If it belonged to $B$, closure across the intervening
transition would put layer $w-1$ or $w+1$ in $L_B$, giving $x$ a
same-labelled point of $B$ at distance 1, contrary to its being
uncovered.

This proves Claim~3. In particular, for every uncovered $x$, each point
in $C_x$ is strictly closer to $x$ than every oppositely labelled point of
$B\cup A$.

\medskip\noindent
\emph{Claim 4.} There is a set $A\subseteq V_{\mathrm{int}}$ of size
$O(2^n\log_2 n/n)$ such that $B\cup A$, with the inherited labels,
represents $f$.

Let $\mathcal U$ be the set of points uncovered by $B$. Each
$x\in\mathcal U$ needs at least one additional prototype chosen from
$C_x$. Choosing such a prototype separately for each $x$ could use as many
as $|\mathcal U|$ additional prototypes. Instead, we exploit the fact that
a single point $y\in V_{\mathrm{int}}$ may belong to $C_x$ for several
different points $x$ and, if selected, can classify all of them correctly
by Claims~2 and~3. We therefore seek a small set
$A\subseteq V_{\mathrm{int}}$ such that every $C_x$, for
$x\in\mathcal U$, contains at least one member of $A$.

We formulate this as a finite set-covering problem. For each
$y\in V_{\mathrm{int}}$, let
\[
 \mathcal S_y=\{x\in\mathcal U:y\in C_x\}.
\]
Thus $\mathcal S_y$ records the uncovered points that Claims~2 and~3
guarantee are classified correctly when $y$ is selected. Choosing the prototype set
$A$ corresponds to choosing the associated family
\[
 \{\mathcal S_y:y\in A\}
\]
and the requirement that every $x\in\mathcal U$ have a selected prototype
in $C_x$ is exactly the requirement that this family cover $\mathcal U$:
\[
 \mathcal U=\bigcup_{y\in A}\mathcal S_y.
\]

Recall that a fractional set cover assigns a non-negative weight to each
available set so that, for every point to be covered, the total weight of
the sets containing that point is at least one. Assign weight $1/a$ to
each $\mathcal S_y$. For $x\in\mathcal U$, the sets containing $x$ are
exactly those indexed by the points $y\in C_x$: equivalently,
$x\in\mathcal S_y$ if and only if $y\in C_x$. Claim~2 gives
$|C_x|\geq a$, and hence the total weight of the sets containing $x$ is
\[
 \sum_{y:\,x\in\mathcal S_y}\frac1a
 =\frac{|C_x|}{a}
 \geq1.
\]
These weights therefore form a fractional set cover. If
$\tau^*$ denotes the minimum total weight of a fractional cover of this
set-cover instance, then
\[
 \tau^*\leq\sum_{y\in V_{\mathrm{int}}}\frac1a
 =\frac{|V_{\mathrm{int}}|}{a}
 \leq\frac{2^n}{a}.
\]

We also need an upper bound on the number of points covered by any one
set. If $x\in\mathcal S_y$, then $y\in C_x$, so
$d_H(x,y)\leq r_x\leq2$. For fixed $y$, there is one cube point at
distance 0 from $y$, there are $n$ at distance 1, and there are
$\binom n2$ at distance 2. Consequently,
\[
 |\mathcal S_y|\leq\Delta,
 \qquad
 \Delta=1+n+\binom n2.
\]
In its set-cover form, the integral-versus-fractional bound of
Lov\'asz~\cite{Lovasz75} states that, if
every available set has cardinality at most $\Delta$, there is an integral
set cover of cardinality at most $(1+\ln\Delta)\tau^*$. Hence there is an
index set $A\subseteq V_{\mathrm{int}}$ such that
$\{\mathcal S_y:y\in A\}$ covers $\mathcal U$ and
\[
 |A|\leq\frac{2^n}{a}(1+\ln\Delta)
 =O(2^n\log_2 n/n).
\]

Give every point of $B\cup A$ its $f$-label. These points classify
themselves at distance 0. If $x\notin B\cup A$ is not uncovered by $B$,
it has a same-labelled prototype in $B$ at distance 1, while Claim~3
places every oppositely labelled prototype at distance at least 2. If
$x$ is uncovered by $B$, the choice of $A$ gives a same-labelled prototype
in $C_x$ at distance at most $r_x$. Claim~3 shows that every oppositely
labelled prototype is strictly further away. Hence $B\cup A$ represents
$f$, proving Claim~4.

Finally, Claims~1 and~4 give
\[
 |B\cup A|
 \leq W_n(f)+T_n+O(2^n\log_2 n/n)
 =W_n(f)+O(2^n\log_2 n/n),
\]
where the last equality follows from \eqref{eq:symmetric-tail-cost}. This
proves the upper bound, with some absolute constant $K$, for all sufficiently
large $n$. Enlarging $K$, if necessary, absorbs the finitely many remaining
cases $n\geq2$, since $0\leq\BNN(f)-W_n(f)\leq2^n$ and
$\log_2 n/n>0$.
\end{proof}

For Example~\ref{ex:symmetric-eight}, the prescribed base has
$L_B=\{0,1,2,6,7,8\}$ and 74 points. Only layer $4$ is uncovered,
with $r_x=1$ and layer-5 supersets as candidates, so a minimum completion
using the candidate sets has 94 prototypes. Removing the redundant tail
prototype $0^8$ gives the optimal 93-point representation.

The error term in Theorem~\ref{thm:symmetric-path-approximation} cannot
be replaced uniformly by $o(2^n/n)$. The following family shows this even
when $W_n$ is at most linear.

\begin{proposition}[A lower bound for the approximation gap]
\label{prop:three-layer-gap}
For $n\geq3$, define $f_n:\cube\to\{0,1\}$ by
$f_n(x)=\lfloor\wt{x}/3\rfloor\bmod2$. Then $W_n(f_n)=n+1$ when
$3\mid n$, and $W_n(f_n)=0$ otherwise. The approximation gap satisfies
\[
 \BNN(f_n)-W_n(f_n)
 \geq\left(\frac23-o(1)\right)\frac{2^n}{n}.
\]
\end{proposition}

\begin{proof}
The layer labels are $000111000111\cdots$, so $D_i=1$ exactly when
$3\mid i$, for $1\leq i\leq n$. No two transition indices differ by two,
so $G_{f_n}$ is edgeless and $\tau_{\mathrm{wt}}(G_{f_n})=0$.
No two transitions are consecutive, and $D_1=0$. Thus the
endpoint convention $D_{n+1}=1$ puts $n$ in $\mathcal T$ exactly when
$3\mid n$, and no other index belongs to $\mathcal T$. Consequently,
$\mathcal H=\{n-1,n\}$ when $3\mid n$, and $\mathcal H=\varnothing$
otherwise. Since $W_n(f_n)=h_n(f_n)$, this gives the stated values of $W_n$.

Let $S$ be any Boolean prototype set representing $f_n$, and write $L_w$
for the $w$th Hamming layer. Consider an internal complete run
\[
 R_u=L_u\cup L_{u+1}\cup L_{u+2},
 \qquad u\equiv0\pmod3,\qquad 3\leq u\leq n-3,
\]
and put $A_u=\binom nu$ and $B_u=\binom n{u+2}$. If $L_u$ or $L_{u+2}$
meets $S$, Part~1 of Lemma~\ref{lem:symmetric-full-layers}, applied at
transition $u$ or $u+3$, respectively, forces the corresponding boundary
layer to lie in $S$. Counting only prototypes inside $R_u$, not those
forced in a neighbouring run, gives $|S\cap R_u|\geq\min\{A_u,B_u\}$.

Suppose instead that neither boundary layer contains a prototype. Since
layers $u-1$ and $u+3$ have the opposite label, the monochromatic component
of a boundary point $x$ is contained in $R_u$. By
Lemma~\ref{lem:monochromatic-shortest-paths}, a closest prototype $p$ to $x$
lies in this component, and hence in $L_{u+1}$.
If $x\in L_u$ and $\supp(x)\not\subseteq\supp(p)$, a shortest path to $p$
can begin by deleting a coordinate in $\supp(x)\setminus\supp(p)$,
entering the oppositely labelled layer $u-1$. This contradicts the lemma's
assertion about every shortest path. Thus $\supp(x)\subseteq\supp(p)$,
and $d_H(x,p)=1$. For $x\in L_{u+2}$, if
$\supp(p)\not\subseteq\supp(x)$, choose a coordinate in
$\supp(p)\setminus\supp(x)$. There is a shortest path to $p$ that adds this
coordinate first, and hence enters layer $u+3$. Thus
$\supp(p)\subseteq\supp(x)$ and again $d_H(x,p)=1$.
The selected middle-layer prototypes must therefore cover all boundary
points by adjacency. Each has $(u+1)+(n-u-1)=n$ neighbours in the two
boundary layers, so $n|S\cap L_{u+1}|\geq A_u+B_u$.
Combining the two cases gives
\[
 |S\cap R_u|
 \geq\min\left\{A_u,B_u,\frac{A_u+B_u}{n}\right\}
 \geq\frac2n\min\{A_u,B_u\}.
\]

Put $h=n^{2/3}$. Then $n/h^2\to0$ and $h/n\to0$, which will give a
negligible binomial tail and uniform comparison of adjacent layer sizes,
respectively, as we now show.

Let $\mathcal U_n$ be the set of multiples $u$ of $3$ such that
$3\leq u\leq n-3$ and $|u-n/2|\leq h+2$. Each $u\in\mathcal U_n$ is the
starting index of the three-layer run
\[
 R_u=L_u\cup L_{u+1}\cup L_{u+2}.
\]
Every layer $L_w$ with $|w-n/2|\leq h$ lies in one of these runs. Indeed, let
$u=3\lfloor w/3\rfloor$. Then $L_w\subseteq R_u$ and
$u\in\{w,w-1,w-2\}$, so
\[
 |u-n/2|\leq |u-w|+|w-n/2|\leq h+2.
\]
For all sufficiently large $n$, this run is internal because $h=o(n)$, and
hence $u\in\mathcal U_n$.

Let $X\sim\operatorname{Bin}(n,1/2)$. Since
\[
 \Pr(X=w)=2^{-n}\binom nw,\qquad
 \mathbb EX=\frac n2,\qquad
 \operatorname{Var}(X)=\frac n4,
\]
Chebyshev's inequality gives
\[
 \sum_{|w-n/2|\leq h}\binom nw
 =2^n\Pr\!\left(\left|X-\frac n2\right|\leq h\right)
 \geq\left(1-\frac{n}{4h^2}\right)2^n.
\]
Since every layer in this central window lies in one of the disjoint runs
indexed by $\mathcal U_n$,
\[
 \sum_{u\in\mathcal U_n}|R_u|
 \geq\sum_{|w-n/2|\leq h}\binom nw
 \geq\left(1-\frac{n}{4h^2}\right)2^n
 =\left(1-\frac{1}{4n^{1/3}}\right)2^n
 =(1-o(1))2^n.
\]

All layer indices in these runs satisfy $|w-n/2|\leq h+4$. Write
$w=n/2+\Delta$, so $|\Delta|\leq n^{2/3}+4$. Then, uniformly throughout
this range,
\[
 \frac{\binom n{w+1}}{\binom nw}
 =\frac{n-w}{w+1}
 =\frac{n/2-\Delta}{n/2+\Delta+1}
 =1+O(n^{-1/3}).
\]
Hence, uniformly for $u\in\mathcal U_n$,
\[
 \binom n{u+1}=(1+o(1))\binom nu,
 \qquad
 \binom n{u+2}=(1+o(1))\binom nu.
\]
Thus the three layers in $R_u$ have asymptotically equal sizes, and
\[
 |R_u|
 =\binom nu+\binom n{u+1}+\binom n{u+2}
 =(3+o(1))\binom nu.
\]
Since $A_u=\binom nu$ and $B_u=\binom n{u+2}=(1+o(1))\binom nu$,
\[
 \min\{A_u,B_u\}=\left(\frac13+o(1)\right)|R_u|
\]
uniformly over $u\in\mathcal U_n$. Summing over the disjoint runs therefore
yields
\[
\begin{aligned}
 |S|&\geq\frac2n\sum_{u\in\mathcal U_n}\min\{A_u,B_u\}\\
 &\geq\frac2n\left(\frac13-o(1)\right)
       \sum_{u\in\mathcal U_n}|R_u|\\
 &\geq\left(\frac23-o(1)\right)\frac{2^n}{n}.
\end{aligned}
\]
Since this holds for every Boolean prototype set $S$ representing $f_n$,
\[
 \BNN(f_n)\geq\left(\frac23-o(1)\right)\frac{2^n}{n}.
\]
Subtracting $W_n(f_n)\leq n+1=o(2^n/n)$ proves the proposition.
\end{proof}

\subsection{The limiting Boolean complexity}

In the uniform symmetric model, $F(x)=\xi_{|x|}$, where
$\xi_0,\ldots,\xi_n$ are independent uniform bits. The change bits
$D_1,\ldots,D_n$ are also independent uniform bits, since the map from
$(\xi_0,\ldots,\xi_n)$ to $(\xi_0,D_1,\ldots,D_n)$ is a bijection.
We first present the real-prototype bound of Kilic, Sima, and Bruck used in
this comparison.

Let $I(F)$ be the number of maximal intervals of consecutive layers on which
$F$ is constant. Equivalently,
\[
 I(F)
 =
 1+\sum_{i=1}^n D_i.
\]
The construction of Kilic, Sima, and Bruck gives
$\NN(F)\leq I(F)$ for every symmetric function \cite{KSB23}. Since
$I(F)-1$ is the number of changes between consecutive layers and every
possible pattern of changes is equally likely,
\[
 I(F)-1\ \sim\ \operatorname{Bin}(n,1/2).
\]
By a standard Chernoff bound \cite{Chernoff52}, with probability tending to
one as $n\to\infty$,
$I(F)=(1/2+o(1))n$. Therefore
\begin{equation}\label{eq:random-symmetric-NN}
 \NN(F)\leq(1/2+o(1))n.
\end{equation}
Their construction uses rational coordinates requiring only
$O(\log_2 n)$ bits each \cite{KSB23}, so the resulting $O(n)$-prototype
representation has polynomial description length. The Boolean complexity has a different scale, as the following theorem shows.

\begin{theorem}[Exact typical Boolean complexity]
\label{thm:random-symmetric-eleven-twentieths}
For each $n$, let $F_n(x)=\xi_{|x|}$, where
$\xi_0,\ldots,\xi_n$ are independent uniform bits. Then
\[
 \frac{\BNN(F_n)}{2^n}\xrightarrow{\Pr}\frac{11}{20},
\]
and
\[
 \frac{\mathbb E\BNN(F_n)}{2^n}\longrightarrow\frac{11}{20}.
\]
Consequently, with probability tending to one,
\[
 \frac{\BNN(F_n)}{\NN(F_n)}
 \geq\left(\frac{11}{10}-o(1)\right)\frac{2^n}{n}.
\]
\end{theorem}

We prove the theorem through three lemmas. Recall that the components of
$G_f$ are paths. For any symmetric function $f$, list the vertices of each
component as
\[
 i_1<i_2<\cdots<i_\ell,
\]
and select $i_2,i_4,\ldots$. This covers every edge and uses
$\lfloor\ell/2\rfloor$ vertices. No cover uses fewer, since the
$\lfloor\ell/2\rfloor$ disjoint edges
$\{i_1,i_2\},\{i_3,i_4\},\ldots$ must each be met. Making this choice in
every component gives a canonical minimum-cardinality vertex cover
$S_{\mathrm{can}}(f)$. Define
\[
 \mathcal L(f)=\mathcal H\cup
   \bigcup_{i\in S_{\mathrm{can}}(f)}\{i-1,i\},
 \qquad
 C_n(f)=\sum_{w\in\mathcal L(f)}\binom nw.
\]
A layer is \emph{selected} when its index belongs to $\mathcal L(f)$.
This rule minimises the number of selected vertices of $G_f$, rather than
their total weight.

For the probabilistic analysis, it is convenient to introduce an auxiliary
two-sided infinite model. Fix $0<\theta<1$, and let
$(D_i)_{i\in\mathbb Z}$ be i.i.d. Bernoulli$(\theta)$ change bits. Call
$i\in\mathbb Z$ an isolated transition if
$D_i=1$ and $D_{i-1}=D_{i+1}=0$. Form a graph on the isolated transition
indices, joining $i$ and $j$ by an edge whenever $|i-j|=2$.

Starting from an isolated transition $i$, the only possible next vertex of
its connected component to the right is $i+2$. Since $i$ is isolated,
$D_{i+1}=0$. Thus $i+2$ is also isolated exactly when $D_{i+2}=1$ and
$D_{i+3}=0$, which has probability $\theta(1-\theta)$. The same argument
applies at each subsequent step, and similarly to the left. Since each
extension has probability at most $1/4$, every component is finite almost
surely.

For each component, choose its second, fourth, and subsequent even-ranked
vertices, exactly as in the canonical cover above. We then declare a layer
position selected if it is adjacent either to one of these chosen vertices
or to a non-isolated transition. This is the two-sided infinite analogue of
the finite selected-layer set $\mathcal L(f)$. By translation invariance,
the probability that $w$ is selected is the same for every
$w\in\mathbb Z$; denote this common probability by $\mu(\theta)$.

\begin{lemma}[Infinite-model density]\label{lem:canonical-density}
In the infinite model with $0<\theta<1$, for every $w\in\mathbb Z$ the
probability that $w$ is selected is
\[
 \mu(\theta)
 =\frac{\theta^2(\theta^2-5\theta+5)}{1+\theta-\theta^2}.
\]
In particular, $\mu(1/2)=11/20$.
\end{lemma}

\begin{proof}
In the two-sided model, let $\mathcal H$ be the set of layer positions
adjacent to a non-isolated transition. The transitions immediately below and above layer
$w$ have indices $w$ and $w+1$, respectively. Thus layer $w$ belongs to
$\mathcal H$ exactly when at least one of these two transitions has another
transition immediately beside it. Equivalently, at least one of
\[
 D_{w-1}D_w,\qquad D_wD_{w+1},\qquad D_{w+1}D_{w+2}
\]
equals one. Each of these three events has probability $\theta^2$.
The intersections of consecutive events have probability $\theta^3$;
the intersection of the first and third, and that of all three, both
have probability $\theta^4$. Inclusion--exclusion gives
\[
 \Pr(w\in\mathcal H)
 =3\theta^2-(2\theta^3+\theta^4)+\theta^4
 =3\theta^2-2\theta^3.
\]

Fix an index $i$. Because two isolated transitions are adjacent when their
indices differ by two, $i$ is the least-indexed vertex of a component
exactly when $i$ is an isolated transition and $i-2$ is not an isolated
transition. The probability that $i$ is isolated is $\theta(1-\theta)^2$,
since this requires
\[
 (D_{i-1},D_i,D_{i+1})=(0,1,0).
\]
Put $u=\theta(1-\theta)$. Conditional on this event, $i-2$ is isolated
precisely when $(D_{i-3},D_{i-2})=(0,1)$, which has probability $u$.
Hence
\[
 \Pr\bigl(i\text{ is the least-indexed vertex of a component}\bigr)
 =\theta(1-\theta)^2(1-u).
\]

If a component has reached an isolated transition $t$, then $D_{t+1}=0$,
and its next possible vertex $t+2$ is isolated exactly when
\[
 (D_{t+2},D_{t+3})=(1,0).
\]
This has probability $u$. Successive extensions involve disjoint pairs
of change bits, also disjoint from the bits determining the beginning of
the component. Thus the extensions are independent of one another and of
the event that the component begins.

A fixed position $j$ is selected as an isolated transition exactly when
it has rank $2r$ in its component for some $r\geq1$. The component must
then begin at $j-2(2r-1)$ and extend $2r-1$ times. These events, as $r$
varies, are mutually exclusive, so
\[
\begin{aligned}
 \Pr(j\text{ is selected as an isolated transition})
 &=\theta(1-\theta)^2(1-u)\sum_{r\geq1}u^{2r-1}\\
 &=\theta(1-\theta)^2(1-u)\frac{u}{1-u^2}
 =\frac{\theta^2(1-\theta)^3}{1+\theta-\theta^2}.
\end{aligned}
\]

The two layer positions associated with a selected isolated transition
$i$ are $i-1$ and $i$. Because $D_{i-1}=D_{i+1}=0$, neither belongs to
$\mathcal H$. Moreover, isolated transition indices cannot be consecutive,
so the layer pairs associated with distinct isolated transitions are
disjoint. A fixed layer position $w$ lies in the pair associated with
transition $w$ or with transition $w+1$. Combining these two disjoint
possibilities with the compulsory-layer probability gives
\begin{equation}\label{eq:canonical-density-formula}
 \mu(\theta)
 =3\theta^2-2\theta^3+
   \frac{2\theta^2(1-\theta)^3}{1+\theta-\theta^2}
 =\frac{\theta^2(\theta^2-5\theta+5)}{1+\theta-\theta^2}.
\end{equation}
At $\theta=1/2$, this is $1/2+2\cdot(1/40)=11/20$.
\end{proof}

\begin{lemma}[Concentration of the canonical layer weight]
\label{lem:canonical-concentration}
Fix $0<\theta<1$. Let $F$ be a symmetric Boolean function generated by
i.i.d. Bernoulli$(\theta)$ change bits $D_1,\ldots,D_n$, with arbitrary
initial label $\xi_0$. Then
\begin{equation}\label{eq:canonical-limit}
 \frac{C_n(F)}{2^n}\xrightarrow{\Pr}\mu(\theta).
\end{equation}
Moreover, if $\mathcal E_n$ is the event that every component of $G_F$
has fewer than $R_n=\lceil3\log_2n\rceil$ vertices, then
$\Pr(\mathcal E_n)\to1$.
\end{lemma}

\begin{proof}
Write $C_n=C_n(F)$ and $\mathcal L=\mathcal L(F)$. By definition,
\[
 \frac{C_n}{2^n}
 =\sum_{w=0}^n\frac{\binom nw}{2^n}\mathbf1_{\{w\in\mathcal L\}}.
\]
The indicator that a layer is selected is not determined by a fixed local
pattern of change bits. Indeed, to decide whether an isolated transition
$j$ is selected by the canonical rule, we must know whether its rank in its
path component is even. Since the canonical cover selects the second,
fourth, and subsequent even-ranked vertices of each component, this may
require looking arbitrarily far to the left to locate the beginning of the
component. We therefore replace the selection rule by a truncated version
that looks back only a fixed number of steps. For this truncated rule,
indicators attached to sufficiently distant layer positions depend on
disjoint sets of change bits and are therefore independent, which allows
the corresponding weighted sum to be shown to concentrate. We then let the
truncation depth grow with $n$, compare the truncated infinite and finite
models away from the endpoints, and finally show that the truncation and
endpoint errors are negligible.

More explicitly, $j$ has rank $1$ if $j-2$ is not isolated, rank $2$ if
$j-2$ is isolated but $j-4$ is not, rank $3$ if $j-2$ and $j-4$ are
isolated but $j-6$ is not, and so on. Thus the rank of $j$ is one plus the
number of consecutive isolated transitions encountered in the sequence
$j-2,j-4,\ldots$ before the first gap.

Fix $L\geq1$ and define an auxiliary truncated decision by examining only
the indices $j-2,j-4,\ldots,j-2L$. If an index that is not an isolated
transition is encountered among them, the rank of $j$ is known and the
truncated decision agrees with the true one. If all $L$ indices are
isolated, assign zero to the truncated decision. This
approximation can differ from the true decision only when the component
extends through all $L$ preceding possible vertices. Conditional on $j$
being isolated, each extension has probability $\theta(1-\theta)$ and
uses a disjoint pair of change bits. The probability of all $L$ extensions
is therefore $[\theta(1-\theta)]^L\leq4^{-L}$.

For each layer position $w$ in the infinite model, let $X_w$ denote the
indicator that layer $w$ is selected under the untruncated selection rule,
and let $X_w^{(L)}$ be the corresponding truncated indicator. By definition,
$\mathbb E X_w=\mu(\theta)$. Determining the truncated decision at a
transition $j$ requires only the change bits with indices from $j-2L-1$
to $j+1$. Under the truncated rule, layer $w$ is selected if either
$w\in\mathcal H$, transition $w$ is selected by the truncated canonical
rule, or transition $w+1$ is selected by that rule. Thus $X_w^{(L)}$
depends only on the change bits with indices from $w-2L-1$ to $w+2$, and
$X_w^{(L)}$ and $X_v^{(L)}$ are independent whenever $|w-v|>2L+3$.
Indeed, if $v>w$ and $v-w>2L+3$, then $w+2<v-2L-1$, so the two indicators
depend on disjoint sets of change bits. Put
\[
 Z_{n,L}=\sum_{w=0}^n b_{n,w}X_w^{(L)},
 \qquad \text{where } b_{n,w}=2^{-n}\binom nw.
\]
This is the truncated infinite-model analogue of $C_n/2^n$. The local
dependence gives
\begin{align*}
 \operatorname{Var}(Z_{n,L})
 &=\sum_{w,v=0}^n b_{n,w}b_{n,v}
   \operatorname{Cov}(X_w^{(L)},X_v^{(L)})\\
 &\leq \sum_{w=0}^n b_{n,w}
   \sum_{|v-w|\leq 2L+3} b_{n,v}\\
 &\leq (4L+7)\max_v b_{n,v}\sum_{w=0}^n b_{n,w}
 =O\!\left(\frac{L}{\sqrt n}\right).
\end{align*}
Here we used $|\operatorname{Cov}(X_w^{(L)},X_v^{(L)})|\leq1$,
$\sum_w b_{n,w}=1$, and
\[
 \max_w b_{n,w}=2^{-n}\binom n{\lfloor n/2\rfloor}
 =O(n^{-1/2}).
\]

The preceding truncation estimate gives
$\Pr(X_w^{(L)}\neq X_w)=O(4^{-L})$, uniformly in $w$, and therefore
\[
 \mathbb E\sum_{w=0}^n b_{n,w}|X_w^{(L)}-X_w|
 =O(4^{-L}).
\]
The quantity inside the expectation is the total binomial weight of the
layer positions on which the truncated and untruncated selection rules
disagree. The constants in this estimate and in the variance bound are
independent of $L$ and $w$.

Take $L=\lceil\log_2 n\rceil$. The variance bound tends to zero, while the
expectation of the weighted disagreement above is $O(n^{-2})$. Chebyshev's
inequality therefore gives
\[
 Z_{n,L}-\mathbb E Z_{n,L}\xrightarrow{\Pr}0.
\]
Since $X_w^{(L)}$ and $X_w$ are $0$--$1$ variables,
\[
 \bigl|\mathbb E X_w^{(L)}-\mu(\theta)\bigr|
 \leq \Pr(X_w^{(L)}\neq X_w)
 =O(4^{-L}),
\]
uniformly in $w$. As $\sum_w b_{n,w}=1$, it follows that
\[
 \mathbb E Z_{n,L}
 =\sum_w b_{n,w}\mathbb E X_w^{(L)}
 =\mu(\theta)+O(4^{-L})
 =\mu(\theta)+o(1).
\]
Finally, Markov's inequality, applied to the weighted disagreement above,
shows that the difference between the truncated and untruncated weighted
sums converges to $0$ in probability. Hence the binomially weighted sum of
the untruncated indicators converges in probability to $\mu(\theta)$.

It remains to transfer this conclusion from the auxiliary infinite model
to the original finite sequence. Let
\[
 \widehat X_w=\mathbf 1_{\{w\in\mathcal L\}}
\]
be the indicator that layer $w$ is selected in the original finite model.
Thus
\[
 \frac{C_n}{2^n}=\sum_{w=0}^n b_{n,w}\widehat X_w.
\]
Let $\widehat X_w^{(L)}$ be the corresponding finite-sequence truncated
indicator. To compare this with the infinite model, extend
$D_1,\ldots,D_n$ to a sequence indexed by all the integers by adjoining
independent Bernoulli$(\theta)$ bits outside the range $1,\ldots,n$. This extended
sequence has the distribution of the auxiliary infinite model. Whenever
the interval from $w-2L-1$ to $w+2$ lies within $1,\ldots,n$, the finite
and infinite truncated indicators use exactly the same change bits, and so
$\widehat X_w^{(L)}=X_w^{(L)}$. A discrepancy can therefore occur only if
\[
 w-2L-1<1\qquad\text{or}\qquad w+2>n.
\]
Since $w$ is an integer, these conditions are equivalent to
$w\leq2L+1$ and $w\geq n-1$, respectively. Thus a discrepancy is possible
only for
\[
 w\in\{0,1,\ldots,2L+1\}\cup\{n-1,n\}.
\]
Since $L=O(\log n)$, the total binomial weight of these layer positions is
\[
 \sum_{w=0}^{2L+1}b_{n,w}
 +\sum_{w=n-1}^n b_{n,w}=o(1).
\]
Thus the finite and infinite truncated weighted sums differ by $o(1)$.

The same truncation estimate shows that the original finite weighted sum
and its truncated version differ by $o(1)$ in probability; the endpoint
layers again contribute only $o(1)$. We have already shown that the
infinite truncated weighted sum converges in probability to $\mu(\theta)$.
Combining these three comparisons gives
\[
 \frac{C_n}{2^n}\longrightarrow\mu(\theta)
 \qquad\text{in probability},
\]
which proves \eqref{eq:canonical-limit}.

We also need a bound on component lengths. For a fixed index $i$, a
component beginning at $i$ and containing at least $R$ vertices requires
$R-1$ successive extensions. The new pairs of change bits are independent
of the event that the component begins, so the probability is at most
$[\theta(1-\theta)]^{R-1}\leq4^{-(R-1)}$. The fixed endpoint bits can
prevent an extension and do not increase this bound. There are at most
$n$ possible least indices. Thus
\[
 \Pr(\text{some component of $G_F$ has at least $R$ vertices})
 \leq n4^{-(R-1)}.
\]
Taking $R=R_n=\lceil3\log_2n\rceil$ gives
$\Pr(\mathcal E_n)\to1$.
\end{proof}

\begin{lemma}[Uniform comparison of canonical and weighted covers]
\label{lem:canonical-weight-comparison}
For every fixed $K_0>0$, there is a sequence $\varepsilon_n(K_0)$ such
that $\varepsilon_n(K_0)\to0$ as $n\to\infty$, and every symmetric Boolean
function $f$ whose graph $G_f$ has all components of size at most
$K_0\log_2 n$ satisfies
\[
 0\leq C_n(f)-W_n(f)\leq\varepsilon_n(K_0)2^n.
\]
\end{lemma}

\begin{proof}
Write $C_n=C_n(f)$ and put $R=\lceil K_0\log_2n\rceil+1$. Since component
sizes are integers, the hypothesis implies that every component has fewer
than $R$ vertices. The canonical selection is a vertex cover of $G_f$.
Since $W_n(f)$ uses a minimum-weight vertex cover, while $C_n(f)$ uses the
weight of this particular canonical cover together with the same
compulsory-layer contribution $h_n(f)$, we have $W_n(f)\leq C_n(f)$. We
bound the difference by comparing the weights within each component in the
central range.

Recall that
\[
 t_i=\binom n{i-1}+\binom ni,
\]
the total size of the two Hamming layers associated with vertex $i$ of
$G_f$. Let $\mathcal C=\{i_1<\cdots<i_\ell\}\subseteq\mathcal I$ be the
vertex set of a component of $G_f$ that meets
\[
 J_n=\{i:|i-n/2|\leq n^{2/3}\}.
\]
Its indices form an arithmetic progression with common difference $2$,
so $i_\ell-i_1=2(\ell-1)<2R$. Since $\mathcal C$ meets $J_n$ and $t_i$
involves layers $i-1$ and $i$, every layer index $w$ occurring in these
weights satisfies
\[
 \left|w-\frac n2\right|
 \leq n^{2/3}+2R
 =n^{2/3}+O(\log n).
\]
Write $w=n/2+s$. Then $|s|\leq n^{2/3}+O(\log n)$, and throughout this
range
\[
 \frac{\binom n{w+1}}{\binom nw}
 =\frac{n-w}{w+1}
 =\frac{n/2-s}{n/2+s+1}
 =1-\frac{2s+1}{n/2+s+1}
 =1+O(n^{-1/3}),
\]
uniformly. The layer indices occurring in the weights indexed by
$\mathcal C$ span at most $2R+1$ consecutive positions. Comparing
successive binomial coefficients across this range therefore gives
\[
 \frac{\max_{i\in\mathcal C}t_i}{\min_{i\in\mathcal C}t_i}
 =1+O(Rn^{-1/3})=1+o(1).
\]

As observed above, the canonical cover $\{i_2,i_4,\ldots\}$ has minimum
cardinality $\lfloor\ell/2\rfloor$.
Define
\[
 T_{\mathrm{can}}(\mathcal C)
 =\sum_{r=1}^{\lfloor\ell/2\rfloor}t_{i_{2r}},
\]
the total vertex weight of the canonical cover, and let
$T_{\mathrm{opt}}(\mathcal C)$ be the minimum total weight of a vertex cover
of this component. If $t_{\min}$ and $t_{\max}$ are the minimum and maximum
vertex weights in the component, then
\[
 T_{\mathrm{opt}}(\mathcal C)
 \geq\lfloor\ell/2\rfloor t_{\min},
 \qquad
 T_{\mathrm{can}}(\mathcal C)
 \leq\lfloor\ell/2\rfloor t_{\max}.
\]
The preceding comparison gives
\[
 T_{\mathrm{can}}(\mathcal C)
 \leq \frac{t_{\max}}{t_{\min}}\,T_{\mathrm{opt}}(\mathcal C)
 =\bigl(1+O(Rn^{-1/3})\bigr)T_{\mathrm{opt}}(\mathcal C).
\]
Hence there is a quantity $\eta_n=O(Rn^{-1/3})$, independent of the
component and tending to zero, such that
\[
 (1-\eta_n)T_{\mathrm{can}}(\mathcal C)
 \leq T_{\mathrm{opt}}(\mathcal C)
 \leq T_{\mathrm{can}}(\mathcal C)
\]
for every component meeting $J_n$. The upper bound holds because the
canonical cover is itself an available vertex cover.

We now sum over the components. A component not meeting $J_n$ has all its
vertex indices outside that central interval. The layers occurring in its
canonical weight therefore satisfy
\[
 \left|w-\frac n2\right|>n^{2/3}-1.
\]
Hoeffding's inequality \cite{Hoeffding63},
applied to both tails of $\operatorname{Bin}(n,1/2)$, gives
\[
 \sum_{|w-n/2|>n^{2/3}-1}\binom nw
 \leq
 2^{n+1}\exp\!\left(-\frac{2(n^{2/3}-1)^2}{n}\right)
 =o(2^n).
\]
Consequently, the sum of the canonical weights of all components not
meeting $J_n$ is $o(2^n)$.

The compulsory-layer contribution $h_n(f)$ is the same in $C_n(f)$ and
$W_n(f)$, so it remains only to compare the vertex-cover contributions.
For each component $\mathcal C$ meeting $J_n$,
\[
 T_{\mathrm{can}}(\mathcal C)-T_{\mathrm{opt}}(\mathcal C)
 \leq \eta_n T_{\mathrm{can}}(\mathcal C).
\]
Summing over all such components gives
\[
 \sum_{\mathcal C\cap J_n\neq\varnothing}
 \bigl(T_{\mathrm{can}}(\mathcal C)-T_{\mathrm{opt}}(\mathcal C)\bigr)
 \leq
 \eta_n\sum_{\mathcal C\cap J_n\neq\varnothing}
 T_{\mathrm{can}}(\mathcal C)
 \leq \eta_n C_n(f)
 \leq \eta_n2^n.
\]

For components not meeting $J_n$, the difference between the canonical and
optimal contributions is at most their total canonical weight. By the
preceding Hoeffding bound, this is at most
\[
 2^{n+1}\exp\!\left(-\frac{2(n^{2/3}-1)^2}{n}\right).
\]
Consequently,
\[
 0\leq C_n(f)-W_n(f)
 \leq \eta_n2^n+
 2^{n+1}\exp\!\left(-\frac{2(n^{2/3}-1)^2}{n}\right).
\]
Dividing by $2^n$, and using
$\eta_n=O(K_0\log_2n/n^{1/3})$, gives
\[
 \frac{C_n(f)-W_n(f)}{2^n}
 \leq O\!\left(\frac{K_0\log_2n}{n^{1/3}}\right)
 +2\exp\!\left(-\frac{2(n^{2/3}-1)^2}{n}\right).
\]
Thus
\[
 \frac{C_n(f)-W_n(f)}{2^n}\leq\varepsilon_n(K_0),
\]
where $\varepsilon_n(K_0)\to0$ as $n\to\infty$, uniformly over all $f$
satisfying the hypothesis. This proves the lemma.

\end{proof}

\begin{proof}[Proof of Theorem~\ref{thm:random-symmetric-eleven-twentieths}]
Lemmas~\ref{lem:canonical-density} and~\ref{lem:canonical-concentration}
give $C_n(F_n)/2^n\xrightarrow{\Pr}11/20$ and
$\Pr(\mathcal E_n)\to1$. On $\mathcal E_n$, every component has fewer than
$3\log_2n$ vertices, since its integer size is less than
$\lceil3\log_2n\rceil$. Thus
Lemma~\ref{lem:canonical-weight-comparison}, with $K_0=3$, gives
$W_n(F_n)=C_n(F_n)+o(2^n)$. The uniform approximation in
Theorem~\ref{thm:symmetric-path-approximation} therefore yields
$\BNN(F_n)/2^n\xrightarrow{\Pr}11/20$. Since this ratio lies in $[0,1]$,
convergence in probability to the constant $11/20$ also implies convergence
of its expectation to $11/20$. Finally,
\eqref{eq:random-symmetric-NN} and the Boolean complexity limit give,
with probability tending to one,
\[
 \frac{\BNN(F_n)}{\NN(F_n)}
 \geq\frac{(11/20-o(1))2^n}{(1/2+o(1))n}
 =\left(\frac{11}{10}-o(1)\right)\frac{2^n}{n}.
\]
\end{proof}

\begin{remark}
Although $F_n$ is symmetric, the definition of $\BNN(F_n)$ imposes no
symmetry condition on the prototype set. The lower-bound argument forces
the complete layers encoded by $\mathcal H$ and by a vertex cover of
$G_{F_n}$, while Theorem~\ref{thm:symmetric-path-approximation} shows that
any additional prototypes needed beyond these forced layers contribute only
$o(2^n)$. These additional prototypes need not themselves form complete
layers. Thus the limit $11/20$ is the asymptotic value of the unrestricted
Boolean nearest-neighbour complexity, not merely of a symmetry-restricted
version of the problem.
\end{remark}

The preceding theorem corresponds to change probability $\theta=1/2$. The same
argument gives a one-parameter extension, showing how the limiting Boolean
complexity varies with the probability of a change between consecutive
Hamming layers.

\begin{corollary}[Independent changes with arbitrary probability]
\label{cor:markov-symmetric-complexity}
Fix $\theta\in[0,1]$. For each $n$, choose $\xi_0$ arbitrarily or at random,
and generate the remaining layer labels by independent change bits
$D_1,\ldots,D_n$ with $\Pr(D_i=1)=\theta$. For the resulting symmetric
function $F$, as $n\to\infty$, for every $\varepsilon>0$,
\[
 \Pr\!\left(\left|\frac{\BNN(F)}{2^n}-c(\theta)\right|
 <\varepsilon\right)\longrightarrow1,
\]
and
\[
 \frac{\mathbb E\BNN(F)}{2^n}\longrightarrow c(\theta),
\]
where
\[
c(\theta)=\frac{\theta^2(\theta^2-5\theta+5)}{1+\theta-\theta^2}.
\]
\end{corollary}

\begin{proof}
For $0<\theta<1$, Lemma~\ref{lem:canonical-density} gives
$\mu(\theta)=c(\theta)$. Lemma~\ref{lem:canonical-concentration} then gives
$C_n(F)/2^n\xrightarrow{\Pr}c(\theta)$ and
$\Pr(\mathcal E_n)\to1$. On $\mathcal E_n$,
Lemma~\ref{lem:canonical-weight-comparison}, with $K_0=3$, gives
$W_n(F)=C_n(F)+o(2^n)$. Theorem~\ref{thm:symmetric-path-approximation}
therefore proves convergence in probability of $\BNN(F)/2^n$, and
boundedness gives convergence of its expectation. At $\theta=0$, the
function is constant, so $\BNN(F)/2^n=2^{-n}\to0=c(0)$. At $\theta=1$,
it is parity or its complement, whose Boolean complexity is $2^n$ by
Proposition~\ref{prop:edge-closure}. Hence $\BNN(F)/2^n=1=c(1)$.
\end{proof}

\section{Open problems}
\label{sec:open-problems}

\paragraph{Counting and $k$-nearest-neighbour representations.}
Between the almost-all lower bound obtained by real-prototype counting
and the universal upper bound of \cite{HLT22} there remains a factor
$O(n)$, and it would be interesting to reduce it. Corollary~\ref{cor:almost-all-BNN}
gives an almost-all lower bound, with constant $c_\ast$, that holds
simultaneously for Boolean $k$-nearest-neighbour representations for every
$k$. The largest constant for which such a simultaneous bound holds is not
known. For $k>1$, the minimum size of a Boolean $k$-nearest-neighbour
representation is also unknown for many particular functions, including
$TH_t^n$.

\paragraph{Covering and monotone extensions.}
For symmetric threshold functions, outside the exact cases of
Section~\ref{sec:thresholds}, the problem is that of determining
$\Cover(n,2r-1,r)$. Even with $q=n-2r+1\geq4$ fixed,
Corollary~\ref{cor:fixed-q} establishes only that
$\BNN(TH_t^n)=\Theta(n)$; its lower and upper bounds give different constant
multiples of $n$, so they do not yield an asymptotic formula.

For a non-constant monotone function $f$,
Theorem~\ref{thm:general-disjunctive-lift} shows that if
\[
 \BNN(f)=\max\{1,c(f)\}+1,
\]
then
\[
 \BNN(z\vee f)=\BNN(f)+1.
\]
This sufficient condition is not necessary, as the example following
Corollary~\ref{cor:orand-parity} shows. It remains to characterise all
non-constant monotone functions $f$ for which
$\BNN(z\vee f)=\BNN(f)+1$.

Theorem~\ref{thm:one-negative} determines the minimum size of a
representation with exactly one negative prototype, whenever such a
representation exists. It does not determine whether this minimum equals
$\BNN(f)$, or whether allowing additional negative prototypes reduces the
total size.

\paragraph{Symmetric functions.}
Let $\Delta_n$ be the maximum of $\BNN(f)-W_n(f)$ over symmetric Boolean
functions on $\cube$. Theorem~\ref{thm:symmetric-path-approximation} and
Proposition~\ref{prop:three-layer-gap} give
\[
 \left(\frac23-o(1)\right)\frac{2^n}{n}
 \leq\Delta_n\leq K\frac{2^n\log_2 n}{n}.
\]
The lower bound is witnessed by the explicit three-layer family in
Proposition~\ref{prop:three-layer-gap}. Thus the uniform approximation
error is optimal up to a logarithmic factor. It remains to determine the
exact order of $\Delta_n$, in particular whether the logarithmic factor
in the upper bound can be removed, and to characterise the symmetric
functions for which $\BNN(f)=W_n(f)$.

On the real-prototype side, Kilic, Sima, and Bruck conjecture that
$\NN(f)=I(f)$ for every symmetric Boolean function $f$, where $I(f)$ is the
number of maximal constant runs of Hamming layers \cite{KSB23}. For a
uniformly random symmetric function $F$, this would imply that $\NN(F)-1$
has the binomial distribution $\operatorname{Bin}(n,1/2)$, and hence that
$\NN(F)=(1/2+o(1))n$ with probability tending to one. Proving even this
typical asymptotic statement would make the real-prototype side of the
comparison in Theorem~\ref{thm:random-symmetric-eleven-twentieths}
asymptotically exact.

\section*{Acknowledgements}

The author is grateful to Gy\"orgy Tur\'an for his feedback.
During the preparation of this work, the author used OpenAI ChatGPT and
Anthropic Claude to assist with literature searches, the
development and checking of some mathematical arguments, and improvements to the
clarity and readability of the manuscript. The author assumes full
responsibility for all content.


\begin{thebibliography}{99}

\bibitem{BKL94}
B. Bollob\'as, Y. Kohayakawa, and T. \L uczak,
\emph{On the Evolution of Random Boolean Functions},
in P. Frankl, Z. F\"uredi, G. Katona, and D. Mikl\'os (eds.),
\emph{Extremal Problems for Finite Sets},
Bolyai Society Mathematical Studies 3,
J\'anos Bolyai Mathematical Society, Budapest, 1994, 137--156.

\bibitem{CG25}
O. \v{C}epek and J. Gli\v{s}i\'c,
\emph{Boolean Nearest Neighbor Language in the Knowledge Compilation Map},
Proceedings of the 22nd International Conference on Principles of Knowledge
Representation and Reasoning, 2025, 240--249.

\bibitem{Chernoff52}
H. Chernoff,
\emph{A measure of asymptotic efficiency for tests of a hypothesis based on
the sum of observations},
The Annals of Mathematical Statistics 23 (1952), 493--507.

\bibitem{CM92}
V. Chv\'atal and C. McDiarmid,
\emph{Small transversals in hypergraphs},
Combinatorica 12 (1992), 19--26.

\bibitem{CMK19}
M. Csik\'os, N. H. Mustafa, and A. Kupavskii,
\emph{Tight Lower Bounds on the VC-dimension of Geometric Set Systems},
Journal of Machine Learning Research 20 (2019), no.~81, 1--8.

\bibitem{DPR25}
M. DiCicco, V. Podolskii, and D. Reichman,
\emph{Nearest Neighbor Complexity and Boolean Circuits},
16th Innovations in Theoretical Computer Science Conference,
LIPIcs 325 (2025), 42:1--42:23.

\bibitem{ES74}
P. Erd\H os and J. Spencer,
\emph{Probabilistic Methods in Combinatorics},
Academic Press, New York, 1974.

\bibitem{FH58}
M. K. Fort, Jr.\ and G. A. Hedlund,
\emph{Minimal coverings of pairs by triples},
Pacific Journal of Mathematics 8 (1958), 709--719.

\bibitem{GJ95}
P. W. Goldberg and M. R. Jerrum,
\emph{Bounding the Vapnik--Chervonenkis Dimension of Concept Classes
Parameterized by Real Numbers},
Machine Learning 18 (1995), 131--148.

\bibitem{GK19}
I. A. D. Gunn and L. I. Kuncheva,
\emph{Bounds for the VC Dimension of 1NN Prototype Sets},
arXiv preprint, 2019,
arXiv:1902.02660.

\bibitem{GKP95}
D. M. Gordon, O. Patashnik, and G. Kuperberg,
\emph{New constructions for covering designs},
Journal of Combinatorial Designs 3 (1995), 269--284.

\bibitem{GS07}
D. M. Gordon and D. R. Stinson,
\emph{Coverings},
in C. J. Colbourn and J. H. Dinitz (eds.),
\emph{Handbook of Combinatorial Designs}, second edition,
Chapman \& Hall/CRC, 2007, 365--372.

\bibitem{HLT06}
P. Hajnal, Z. Liu, and G. Tur\'an,
\emph{Nearest neighbor representations of Boolean functions},
9th International Symposium on Artificial Intelligence and Mathematics,
Article P44, 8 pp., AAAI Press, 2006.

\bibitem{HLT22}
P. Hajnal, Z. Liu, and G. Tur\'an,
\emph{Nearest neighbor representations of Boolean functions},
Information and Computation 285 (2022), 104879.

\bibitem{Hoeffding63}
W. Hoeffding,
\emph{Probability inequalities for sums of bounded random variables},
Journal of the American Statistical Association 58(301) (1963), 13--30.

\bibitem{IKI94}
M. Inaba, N. Katoh, and H. Imai,
\emph{Applications of weighted Voronoi diagrams and randomization to
variance-based $k$-clustering},
Proceedings of the Tenth Annual Symposium on Computational Geometry,
332--339, Association for Computing Machinery, 1994.

\bibitem{JLR00}
S. Janson, T. \L uczak, and A. Ruci\'nski,
\emph{Random Graphs},
Wiley-Interscience Series in Discrete Mathematics and Optimization,
Wiley, New York, 2000.

\bibitem{Keevash14}
P. Keevash,
\emph{The existence of designs},
preprint,
arXiv:1401.3665
(2014; latest revision 2024).

\bibitem{KSB23}
K. M. Kilic, J. Sima, and J. Bruck,
\emph{On the Information Capacity of Nearest Neighbor Representations},
2023 IEEE International Symposium on Information Theory,
1663--1668, 2023.

\bibitem{KSB24C}
K. M. Kilic, J. Sima, and J. Bruck,
\emph{Nearest Neighbor Representations of Neural Circuits},
2024 IEEE International Symposium on Information Theory,
3077--3082, 2024.

\bibitem{KSB24N}
K. M. Kilic, J. Sima, and J. Bruck,
\emph{Nearest Neighbor Representations of Neurons},
arXiv preprint, 2024,
arXiv:2402.08748.

\bibitem{Lovasz75}
L. Lov\'asz,
\emph{On the ratio of optimal integral and fractional covers},
Discrete Mathematics 13 (1975), 383--390.

\bibitem{PW25}
Y. Polyanskiy and Y. Wu,
\emph{Information Theory: From Coding to Learning},
Cambridge University Press, 2025.

\bibitem{Rodl85}
V. R\"odl,
\emph{On a packing and covering problem},
European Journal of Combinatorics 6 (1985), 69--78.

\bibitem{Saposhenko75}
A. A. Sapozhenko,
\emph{Geometric structure of almost all Boolean functions},
Problemy Kibernetiki 30 (1975), 227--261 (in Russian).

\bibitem{Warren68}
H. E. Warren,
\emph{Lower bounds for approximation by nonlinear manifolds},
Transactions of the American Mathematical Society 133(1) (1968), 167--178.

\bibitem{Weber83}
K. Weber,
\emph{Subcubes of random Boolean functions},
Elektronische Informationsverarbeitung und Kybernetik
19 (1983), 365--374.

\end{thebibliography}
\end{document}